\documentclass[aps,apl,twocolumn, notitlepage, superscriptaddress]{revtex4-2}

\usepackage[utf8]{inputenc}
\usepackage{amsmath, amssymb, amsthm}
\usepackage{xcolor, framed}
\usepackage{physics, siunitx, bm}
\usepackage{graphicx}
\usepackage{mathtools}
\usepackage[colorlinks=true, linkcolor=blue, citecolor=blue,
filecolor=blue, runcolor=blue, urlcolor = blue]{hyperref}
\usepackage{slashed} 
\usepackage[caption=false]{subfig}

\usepackage{import}
\usepackage{xifthen}

\usepackage{algorithm}
\usepackage{algpseudocode}

\newtheorem{thm}{Theorem}

\DeclareMathOperator{\Ex}{E}

\newcommand{\mb}[1]{\mathbf{#1}}

\newcommand{\bsy}[1]{\boldsymbol{#1}}

\colorlet{shadecolor}{gray!25}

\definecolor{issuePJA_color}{rgb}{1.0,0.0,0.0}

\definecolor{commentPJA_color}{rgb}{1.0,0.0,0.8}

\definecolor{issuePT_color}{rgb}{0.0,0.7,0.0}

\definecolor{commentPT_color}{rgb}{0.1,0.1,0.6}

\definecolor{rev_color}{rgb}{0.0,0.0,0.0}
\newcommand{\rev}[1]{{\textcolor{rev_color}{#1}}}

\begin{document} 

\title{Protein
Drift-Diffusion Dynamics and  
Phase Separation in \\ Curved Cell Membranes and
Dendritic Spines: \\ Hybrid Discrete-Continuum Methods}
\author{Patrick D. Tran} \email{patrickduytran@gmail.com} 
\affiliation{Physics, College of Creative Studies, University of 
California, Santa Barbara (UCSB)}
\author{Thomas A. Blanpied} 
\affiliation{Department of Physiology, University of Maryland}
\author{Paul J. Atzberger} 
\email{atzberg@gmail.com} 
\affiliation{Department of Mathematics and Mechanical Engineering, 
University of California, Santa Barbara (UCSB),
} 

\begin{abstract} 
We develop methods for investigating protein drift-diffusion dynamics in
heterogeneous cell membranes and the roles played by geometry, diffusion,
chemical kinetics, and phase separation.  Our hybrid stochastic numerical
methods combine discrete particle descriptions with continuum-level models
for tracking the individual protein drift-diffusion dynamics when coupled to
continuum fields.  We show how our approaches can be used to investigate
phenomena motivated by protein kinetics within dendritic spines.  The spine
geometry is hypothesized to play an important biological role regulating
synaptic strength, protein kinetics, and self-assembly of clusters.  We
perform simulation studies for model spine geometries varying the neck size
to investigate how phase-separation and protein organization is influenced by
different shapes.  We also show how our methods can be used to study the
roles of geometry in reaction-diffusion systems including Turing
instabilities.  Our methods provide general approaches for investigating
protein kinetics and drift-diffusion dynamics within curved membrane
structures.
\end{abstract}

\maketitle

\section{Introduction} 
In cellular biology, the morphological shapes of cell
membranes play important roles in protein transport and 
kinetics.  Cell
membranes often take on shapes having characteristic geometries
or topologies associated with biological 
function~\cite{Alberts2014,BassereauRoux2005,
Noske2008,Adler2010,BassereauShapeProteinDistribution2013,Vogel2006,
Kusters2014}. Membranes arising in cell biology consist of
heterogeneous mixtures of lipids, proteins, and other small
molecules~\cite{Alberts2014}.  The individual and collective dynamics of
membrane associated molecules carry out diverse functions in cellular processes
ranging from signaling to motility~\cite{Alberts2014,Voeltz2007,Groves2007,
Powers2002,Muller2012,NelsonStatMechMem2004}.  Membranes are effectively
two dimensional fluid-elastic structures resulting in processes that 
can be significantly different than their counter-parts occuring in 
bulk three dimensional fluids~\cite{Saffman1975,
Peters1982,BassereauMobilityConfinement2011,AtzbergerBassereau2014}.
Investigating such cellular processes using computational simulation
requires the ability to capture these effects and the  
geometric and topological contributions of curved membrane
structures to protein drift-diffusion dynamics and kinetics.

We introduce computational methods based on a hybrid approach
coupling discrete and continuum descriptions.
For low concentration species, we track individual proteins as 
discrete particles. For other species, we track 
contributions using continuum fields.  We
circumvent many of the challenges of differential geometry and directly
approximating surface PDEs by developing discrete localized models that capture
geometric effects.  In this way, behaviors of our model emerge on larger
length-scales in a manner capturing the relevant underlying physical phenomena,
while avoiding some of the more common challenges associated with direct
application of PDE discretizations and differential geometry.  We also provide
stochastic local models to account for discrete effects and other fluctuations.

Many computational methods have been introduced for studying membranes.
Methods modeling at the level of continuum fields and partial differential
equation (PDE) descriptions include continuum concentration and 
phase fields in vesicles in~\cite{Allard_Lowengrub_2017}, protein aggregation
in~\cite{Rangamani2021}, and phase separation in~\cite{Yushutin2019}.
Methods modeling at the level of particles include Monte-Carlo (MC) Methods and
Kinetic MC (KMC) in~\cite{Pastor1994,Sintes1998,Kerr2008,Schoeneberg2014,Kahraman2016,
Atzberger_Isaacson_2014,Collins2010,Crane2020}, Molecular Dynamics (MD)
studies in~\cite{Tieleman1997, Goossens2018,Grouleff2015}, and Coarse-Grained (CG)
Models in~\cite{DesernoVirusAggregation2007,Marrink2007}.  Some work has been
done on hybrid discrete-continuum approaches for membranes 
in~\cite{Camley2010,AtzbergerNaji2009,
AtzbergerBilayerVesicle2013,AtzbergerNaji2009,
AtzbergerSigurdsson2012,Reister2005,ReisterGottfried2010, Oppenheimer2009}, and
taking into account geometric effects
in~\cite{AtzbergerSigurdsson2012,Atzberger_Surf_Fluct_2019}, and through 
point-cloud representations in~\cite{Macdonald2013,
AtzbergerGrossHydroSurf2018,Lai2017, Atzberger_Surface_FPT_2021,Crane2020}.

\rev{Our methods provide new ways for handling hybrid discrete-continuum stochastic descriptions for protein drift-diffusion dynamics and kinetics within heterogeneous curved membranes.  We focus particularly on the roles played by phase separation, geometry, and fluctuations arising from discrete number effects.  We develop stochastic methods for bidirectional coupling between protein dynamics and evolving phase fields.  For capturing geometric contributions for general shapes, we also introduce numerical approaches using the induced metric from the embedding space which allows for avoiding the need for potentially cumbersome calculations using explicit expressions from differential geometry.  We also introduce methods for tracking proteins both at the level of continuum mean-field concentration fields and beyond mean-field theory at the level of discrete individual proteins with stochastic trajectories.
}

\rev{Our work is motivated by understanding protein interactions within cell curved membranes, such as neuronal dendritic spines.  Dendritic spines are small $\sim 500$nm structures that 
are attached along the larger shaft of the dendrites of
neurons~\cite{Alvarez2007,Yuste1995}.  These are critical structures in the brain mediating the input of synaptic communications between neurons.  
Dendritic spines have unique morphologies that physiologically change in molecular composition, size, and shape to modulate the synaptic strength between neurons as part of long-term learning and memory~\cite{BlanpiedCaMKII2014,Yuste1995,Herring2016,Nicoll2017}.  
The important connection between membrane geometry and synaptic strength is an active area of current theoretical and experimental research \cite{Holcman2011,Li2016,Wang2016,
Adrian2017,Chen2017,
Simon2014,Cartailler2018,Rangamani2019,Borczyk2019,Tapia2019,Miermans2017,
Toennesen2016,Nishiyama2015}.  
Recent advances in microscopy and single-particle tracking
techniques are providing some indicators of the underlying processes
regulating protein transport and kinetics within spines \cite{Heine2008,
Blanpied2012, BlanpiedCaMKII2014,Blanpied2010, Holcman2015,Jaskolski2009,
Toennesen2014,Groc2020}. 
}

\rev{A hypothesis which we explore with our methods is that protein complexes
can locally nucleate liquid-liquid phase separations that couple to proteins to influence both diffusive transport and kinetics.  Our work is motivated by the recent observations concerning the roles of SynGAP binding to PSD-95 resulting in complexes that phase separate playing a role in organizing proteins in
dendritic spines in forming the post-synaptic density~\cite{Zeng2016,Zeng2019}. In
our initial work presented here, we do not commit in models yet to specific
proteins, but focus more generally on mechanisms by which protein diffusion and
phase separation can drive cluster formation and augment reaction kinetics
consistent with such observations.  For exploring such hypotheses, we focus on
the development of methods for quantitative biophysical models and computational
simulations capable of investigating such effects.  For this purpose, we
develop spatial-temporal biophysical models using our introduced hybrid discrete-continuum approach to investigate the roles of geometry and phase separation in protein diffusion and kinetics.}

\rev{We discuss details for our modeling approach for 
protein drift-diffusion dynamics in Section~\ref{sec_drift_diffusion}
and our simulation studies in Section~\ref{sec_simulations}.
We discuss the drift-diffusion dynamics of proteins coupled to 
phase fields and
discretizations for curved membranes having general shapes in Sections~\ref{sec_particle_dynamics}, ~\ref{sec_markov_chain}, ~\ref{sec_detailed_balance}, and ~\ref{sec_reaction_coupling}.
We investigate the accuracy of our numerical approximations in Section~\ref{sec_approx_err}.  
We demonstrate how our methods can be used to compute first passage times 
and other statistics for membranes of general shape, in some cases without
the need for costly Monte-Carlo sampling, in Section~\ref{sec_u_stat}.  
}

\rev{We perform simulations using our methods of discrete and continuum systems
in~\ref{sec_simulations}.  We develop continuum field methods for 
reaction-diffusion systems with pattern forming Turing instabilities 
capturing the contributions of geometry in Section~\ref{sec_rd}. 
We develop mechanistic models for dendritic spine shapes 
and perform simulation to investigate the roles of 
geometric effects, heterogeneities arising from
phase separation, and discrete number effects 
in Section~\ref{sec_spine} and Section~\ref{sec_clustering_phase_sep}.
The introduced methods provide general approaches for
investigating protein transport and kinetics within heterogeneous cell
membranes.}

\section{Hybrid Discrete Particle-Continuum Approach for Curved Surfaces} 
\label{curved_diffusion}
\label{sec_drift_diffusion}
\rev{We develop models using a hybrid discrete particle and 
continuum field approach for investigating transport 
and reactions within curved surfaces, see
Figure~\ref{fig_markov_chain}.}
For low concentration chemical
species, we introduce approaches for modeling at the level of tracking 
individual particles.  For larger concentrations, we develop 
continuum field descriptions and methods for curved surfaces.  
We also account for the bi-directional coupling between 
the discrete particles and continuum fields.

\begin{figure} 
\centering 
\includegraphics[width=0.45\textwidth]{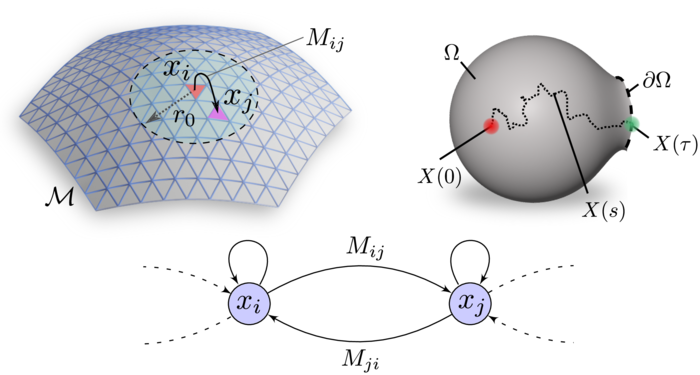}
\caption{Markov-Chain Discretization of Particle Drift-Diffusion on Curved 
Surfaces.  The surface is discretized using a triangulation with random 
walks between centriods $x_i$ and $x_j$ having jump rates $M_{ij}$ 
determined by the local geometry (top-left) which can be expressed as 
a Markov-Chain (bottom).  
The Markov-Chain discretization yields a backward equation for $\Omega$ 
allowing readily for computation of first-passage times $\tau$
(FPTs) for reaching the boundary $\partial \Omega$ and other 
statistics without the need in some cases for Monte-Carlo sampling
(top-right).  } 
\label{fig_markov_chain} 
\end{figure}

\subsection{Particle Drift-Diffusion Dynamics} 
\label{sec_particle_dynamics}

We discuss a few results related to the drift-diffusion
dynamics of particles useful in developing our methods.  \rev{We then 
approximate the particle dynamics using a Markov-Chain process~\cite{Ross1996} 
with jump rates based on estimating the local geometry, see
Figure~\ref{fig_markov_chain}.}

\rev{The drift-diffusion dynamics of a particle immersed within a viscous fluid
is given by the Langevin equation~\cite{Gardiner1985,Reichl1997}
\begin{equation} 
m\,\dd\mb{V}_t= -\gamma \mb{V}_t\, dt -\nabla U(\mb{X}_t)\,\dd
t+\sqrt{2k_BT\gamma}\,\dd \mb{W}_t, \label{langevin}
\end{equation} 
where $d\mb{X}_t = \mb{V}_tdt$.  This is to be interpreted as 
an It\^o process~\cite{Oksendal2000,Gardiner1985}.  The
$\gamma$ is the drag,
$U(\mb{x})$ is a potential energy,  $k_B$ is Boltzmann's
constant, and $T$ is the temperature~\cite{Reichl1997}.
The $\dd \mb{W}_t$ are increments of 
the Wiener process~\cite{Oksendal2000,Gardiner1985}.  The diffusion
coefficient is given by $D = k_B{T}/\gamma$.}

\rev{We remark that in some cases the protein diffusion can also be influenced
by local protein-induced deformations of the
membrane~\cite{AtzbergerBassereau2014}, crowding effects (sub-super diffusive
and other regimes)~\cite{Hoefling2013,Fanelli2010}, or cytoskeletal
interactions~\cite{Takenawa2007,Kahraman2016}.  Our methods could be combined with such models for local deformation and crowding by augmenting the 
Markov-Chain jump rates.  In the present work, we treat
the membrane as having on average an effectively constant shape.  We focus on
models in the normal diffusive regime with explicit modeling of protein
species, binding partners, or obstacles.  We treat additional contributions
through effective diffusion coefficients or in the formulated protein
interaction forces.}

When $m/\gamma \ll \ell^2/D$,
where $\ell$ is the radius of the particle, the inertial contributions
are negligible.  In this regime, the Langevin equation can be reduced
to the over-damped Smoluchowski equation
\begin{equation} 
\dd\mb{X}_t=-\frac{1}{\gamma}\nabla U(\mb{X}_t)\,\dd t+\sqrt{2D}\,\dd\mb{W}_t.  
\end{equation}
This can be expressed in terms of probability densities 
$\rho(\mb{x},t)$ for when $\mb{X}_t = x$. This satisfies the 
Fokker-Planck (FP) equation 
\begin{equation}
\pdv{\rho}{t}=-\nabla\cdot \mb{J},\;\;\;\; \;\;
  \mb{J} = \left(-\frac{1}{\gamma}\nabla U\right)\rho - D\nabla \rho.\label{FP}
\end{equation} 
When $U = 0$, this has the well-known Green's Function for Euclidean space
\begin{equation}
K(\mb{x}',\mb{x};t)=\frac{1}{(4\pi Dt)^{d/2}}
\exp\left(-\frac{(\mb{x}' - \mb{x})^2}{4Dt}\right).
\label{kern}
\end{equation} 
From the FP equation, the $K(\mb{x}',\mb{x};t)$ has 
the interpretation of the probability of a 
particle starting with $\mb{X}_0 = \mb{x}$ and diffusing to location 
$\mb{X}_t = \mb{x}'$ over the time duration $t$.  We will use a 
related approach to determine
our Markov-Chain jump rates.  The FP equation has
a steady-state $\rho^*$ with detailed balance if 
\begin{equation}
\mb{J}[\rho^*] = -\gamma^{-1}\nabla U\rho^* - D\nabla \rho^* = 0.
\end{equation}
For a smooth $U(x)$ with a sufficient growth
rate as $|x| \rightarrow \infty$, the FP equation has 
steady-state with detailed-balance for the distribution
\begin{equation}
\rho^*(\mb{x}) = \frac{1}{Z} \exp\left(-\frac{U(x)}{kT}\right).
\end{equation} 
This is the Gibbs-Boltzmann distribution, and $Z$ is the partition function
normalizing this to be a probability density~\cite{Reichl1997}.  

\subsection{Markov-Chain Discretization for Particle Drift-Diffusion 
Dynamics on Curved Surfaces}
\label{sec_markov_chain}

We model the drift-diffusion dynamics of individual particles on curved
surfaces using discrete Markov-Chains with the jump rates based on 
the local geometry.  The surface is discretized 
into a triangulated mesh and each particle is tracked by the triangle 
which it occupies. We use that the surface metric is induced by the 
surrounding embedding space~\cite{Pressley2001,Abraham1988}.
We use for the local jump rates  
\begin{equation}
M_{ij} = C_i\exp\left(-\frac{|x_{i}-x_{j}|^{2}}{\epsilon^2}\right).
\label{diffprob} 
\end{equation} 
To approximate the diffusion over the time-scale $\Delta{t}$ on the 
surface $\mathcal{M}$, we use $\epsilon = \sqrt{4D\Delta{t}}$.  The 
$x_i \in \mathcal{M}$ are the centers of the surface triangulation.
The $C_i$ denotes the normalization constant when summing over 
index $j$ ensuring that $M$ is a right stochastic matrix~\cite{Ross1996}. 

\rev{We remark in our methods the approach of using the induced metric from the
embedding allows for capturing the geometric contributions while avoiding the
potentially cumbersome expressions that can arise from a more explicit
treatment using differential geometry.  Our methods utilize that the distances
between points on the surface correspond to the arc-length determined by the
path as measured using the embedding space and its associated notion of
distance.  As a result, in our methods the discretized surface inherits its
metric without the need for further analytic derivations.  Further properties
of our methods include that the discretizations for curved surfaces are based
on Markov-Chain transitions and as a consequence will have mass conservation up
to numerical round-off errors.}

The kernel $M_{ij}$ has been shown in the limit of 
refining the surface sampling to approximate 
diffusion under the surface Laplace-Beltrami 
operator $D\Delta_{\mathcal{M}}$~\cite{Singer2006, Belkin2008}.
This has been shown to have the accuracy
\begin{equation}
\begin{split}
\lim_{N\to\infty}\sum_{j=1}^{N}(M_{ij}-I_{ij})u_j=
\frac{\varepsilon^{2}}{4}\Delta_{\mathcal{M}}u(x_{i})
+\mathcal{O}\left(\frac{1}{N^{1/2}},\varepsilon^{4}\right). 
\label{equ_M_conv_LB}
\end{split} 
\end{equation}
The $\mathcal{O}$ holds with $\epsilon \rightarrow 0$.
The $N$ is the number of points sampling $\mathcal{M}$ subject to 
a uniformity condition~\cite{Singer2006,Belkin2008,
AtzbergerGrossHydroSurf2018}.
The $u_j = u(x_j)$ samples a smooth test function $u$.  
Intuitively, this follows since the stochastic matrix 
$M$ converges to the operator as
$\exp(\varepsilon^{2}\Delta_{\mathcal{M}}/4)=\exp(D\Delta{t}\Delta_{\mathcal{M}})\simeq
\exp(D\Delta{t}\Delta)$ and 
$\exp(D\Delta{t}\Delta_{\mathcal{M}})
\simeq I+D\Delta{t}\Delta_{\mathcal{M}} 
\simeq I+D\Delta{t}\Delta$, where $\Delta = \nabla\cdot\nabla$ is the standard
Laplacian. The last two 
terms are motivated by the Taylor expansions of the exponential
and the geometric terms as $\epsilon \rightarrow 0$.  

\begin{figure} 
\centering
\includegraphics[width=0.40\textwidth]{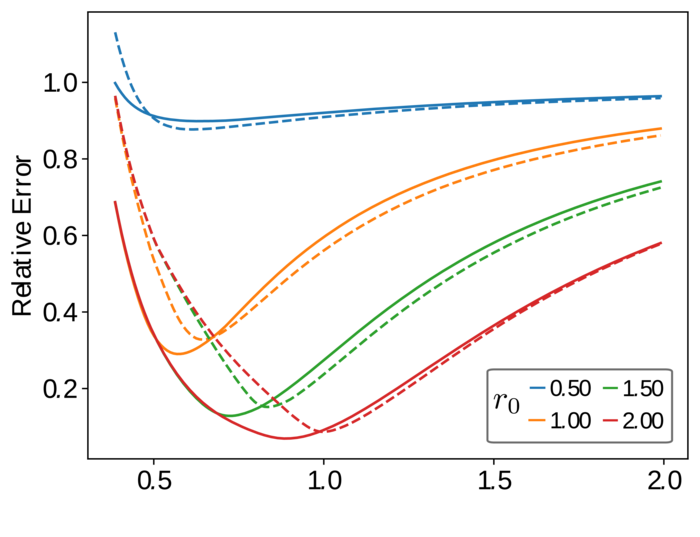}
\caption{Approximation of Laplace-Beltrami for $\epsilon$ and $r_0$.
Shown is the relative error of the Markov-Chain surface discretization
compared with the Laplace-Beltrami operator on the unit sphere $S^2$.  
The case without area correction in equation~\eqref{equ:M_db} are shown
as solid lines and with the area correction with dotted lines.
There is a trade-off in $\epsilon$ to make it sufficiently large
to contain enough points sampling the surface while also 
maintaining locality.  For the truncation radius $r_0 \geq 1.5$ 
the optimum is around $\epsilon \simeq 1$.}
\label{fig_conv_rel} 
\end{figure}

Since the
the surface properties will only be approximated when there
are a sufficient number of sample points in the support of the 
kernel, for a given $N$ there is a trade-off in the choice 
of $\epsilon$.  If $\epsilon$ is too large the approximation
will not be of local surface properties.  If $\epsilon$ is 
too small only the center point will contribute significantly
to the kernel.  We investigate further this trade-off in 
$\epsilon$ and the resulting approximation accuracy in 
Section~\ref{sec_approx_err}.

\rev{We remark that our approach can also readily be modified to obtain models
for proteins in super-sub diffusive regimes.  One way this can be accomplished
is by constructing random walks that are non-local on the mesh having
long-tails or making state transitions between additional non-spatial states.
These models could be constructed from theoretical considerations or empirical
observations of autocorrelation statistics~\cite{Hoefling2013}.}

\subsection{Detailed Balance and Area Corrections}
\label{sec_detailed_balance}
To incorporate the contributions of the drift  arising
from $U$ in equation \eqref{langevin}, we 
consider the Gibbs-Boltzmann distribution expressed
as 
$\rho(x) = {Z}^{-1} \exp\left(-\beta U(x)\right)$, with
$\beta = (k_B{T})^{-1}$ the inverse thermal energy.
We seek discretizations preserving statistical 
structure, such as detailed-balance as in 
~\cite{Schutte2011,Elston2003}. For our curved surfaces,
we seek discretizations for methods that have at steady-state 
the surface Gibbs-Boltzmann distribution with 
detailed balance~\cite{Reichl1997}.
We can express the evolution of the discrete probability in
terms of net fluxes as
\begin{equation} 
\begin{split}
  p_i^{(n+1)} = p_i^{(n)} + \sum_{j,j \neq i} J_{ij}^{(n)},
  \hspace{0.4cm} J_{ij}^{(n)} =  p_j^{(n)} M_{ji}
  -p_i^{(n)}M_{ij}.
\end{split} 
\end{equation} 
At steady-state $p_k^{(n+1)} = p_k^{(n)}=p_k^*$ we design our transition rates
so that the we have a discrete surface Gibbs-Boltzmann distribution with 
approximate detailed balance.   The discrete detailed-balance 
$J_{ij} = 0$ gives the conditions
\begin{equation} 
p_i^* = \exp\left(-\beta U_i\right) A_i/Z, \hspace{0.1cm} 
\frac{M_{ij}}{M_{ji}} = \frac{p_j^*}{p_i^*} 
= \frac{\exp\left(-\beta U_j\right)A_j}{\exp\left(-\beta U_i\right)A_i},
\label{equ:M_db}
\end{equation} 
where $U_k = U(x_k)$ and $Z = \sum_i \exp\left(-\beta U_i\right)A_i$.

Motivated by these conditions, we discretize using the transition rates
\begin{equation} 
\begin{split}
  M_{ij} &= C_i  \exp\left(-\frac{|x_i-x_j|^2}{\varepsilon^2}\right)
  \times\exp\left(\frac{U(x_i)-U(x_j)}{2k_BT}\right) \\
  &\times\exp\left(\frac{1}{2}\ln\frac{A_j}{A_i}\right), 
\label{equ:M}
\end{split}
\end{equation} 
where $C_i$ normalizes the $M_{ij}$ to be a probability when
summing over the index $j$.
The conditions hold up to the normalization ratio $C_i/C_j \rightarrow 1$ as 
the discretization is refined with 
$N \rightarrow \infty$, $\epsilon \rightarrow 0$.

\begin{algorithm}[H]
\caption{Hybrid Reaction-Diffusion Method}
\label{algrdmethod}
\begin{algorithmic}[0]
\State $X_i^0 \gets $ initialize particle positions.
\State $S_i^0 \gets $ initialize particle states.
\State $q^0 \gets $ initialize fields.
\State $n \gets 0$.
\While{n $< $ numsteps} \\
$\triangleright$ Update particles using jump rates from Section~\ref{sec_detailed_balance}.\\
\hspace{0.2cm} $X_i^{n+1} \gets x_b$, 
$\Pr{X_i^{n+1} = x_b| X_i^{n} = x_a} =  M_{ba}(q^n,\{S_k^n\})$. \\
$\triangleright$ Check for chemical reactions $\mathcal{C}$ and update the states.\\
\hspace{0.2cm} $S_i^{n+1} \gets s_a$, 
$\Pr{S_i^{n+1} = s_a| X_i^{n+1} = x_b, \{X_k^{n}\}_{k\neq i}, \{S_k^n\}}$. \\
$\triangleright$ Update the fields.\\
\hspace{0.2cm}  $q^{n+1} \gets \Psi(q^n)$. \\
\hspace{0.2cm}  n $\gets$ n + 1.
\EndWhile
\end{algorithmic}
\end{algorithm}

\subsection{Chemical Reactions and Field Coupling}
\label{sec_reaction_coupling}
\rev{
We also formulate methods for capturing chemical reactions.  We introduce for each particle $X_i(t)$ an associated "state" variable $S_i(t)$. The $S_i(t)$ can be thought of tracking the chemical species to which $X_i$ belongs at time $t$.  We consider a collection of chemical reactions $\mathcal{C} = \{c_\ell\}_{\ell =1}^m$. For $c_\ell$ we model the reaction using a Smoluchowski reaction radius $r_\ell$ and probability $p_\ell$.  For a second-order reaction, this means if two particles $X_i^n$
and $X_j^n$ have $\|X_i^n - X_j^n\| < r_\ell$, then we take probability $p_\ell$ for a reaction to occur.  For first-order reactions, we just take $p_\ell = p_\ell(\Delta{t})$ for each particle for the 
probability a reaction occurs over the time-step.
}

\rev{
We also allow for coupling of the particles to underlying fields $q = q(\mb{x},t)$.  We discretize $q$ in time and space using the lattice with $q_a^n \sim q(\mb{x}_a,t_n)$.
As we discuss in later sections, the 
fields can influence the local jump rates
and reactions.  The evolution of the field is modeled using $q^{n+1} = \Psi(q^n)$.  The case of multiple fields can be handled by taking $q$ to be 
vector-valued.  We shall discuss in more detail specific evolution models for concentration fields in Section~\ref{sec_rd} and phase-fields in Section~\ref{sec_phase_sep}.  We give a summary of the steps in our method in Algorithm~\ref{algrdmethod}.
}

\subsection{Approximation Errors}
\label{sec_approx_err}
\rev{To obtain accurate diffusive dynamics we consider how parameter choices influence the jump rates and the approximation of the Laplace-Beltrami operator discussed in Section~\ref{sec_detailed_balance}.}
We investigate the accuracy of the surface discretizations
and the trade-offs in the choice of $\epsilon$ between sufficient
sampling and maintaining locality.  We perform our studies
for the surface of the unit sphere $S^2$, which has the Laplace-Beltrami 
operator $\Delta_{\mathcal{M}}$ with eigenfunctions 
corresponding to the spherical 
harmonics~\cite{Mueller1966,
HandBookSphericalHarmonics2010,Abraham1988}.
In the comparisons, our numerical operator $M$ is obtained
from equation~\eqref{equ:M} and the scalings
indicated in~\eqref{equ_M_conv_LB} to yield the approximation
$L = (M - I) \simeq (\epsilon^2/4)\Delta_{\mathcal{M}}$.
In practice for efficient calculations, the $M_{ij}$ is constructed by
truncating the kernel only to use neighbors $x_j$ within the distance $r_0$
with $|x_i - x_j| < r_0$.
We consider the errors for a test function $v$ given by 
\begin{equation}
e_{\text{LB}}[v] = \left(\frac{4}{\epsilon^2}\left(M - I\right)-\Delta_{S^2}\right)v.
\end{equation} 

\rev{In Figure~\ref{fig_conv_rel}, we show the relative errors based on $\epsilon_{rel} = \|e_{\text{LB}}[v]\|_1/\|\Delta_{S^2}[v]\|_1$, where $\|\cdot\|_1$ is the $L^1$-norm averaging over the surface. We consider case when $v$ is the spherical harmonic corresponding to
$v = v(x,y,z) = z$ restricted to the surface.  The sphere is discretized using a triangular mesh with nearly uniform elements having $N=10,0000$ nodal points.
We consider how the error varies for different choices of 
$\epsilon \in [0.25,2.0]$ and $r_0 \in[0.5,2.0]$. 
Letting $\delta{x} = \min_{ij} |x_i - x_j|$, we 
find that when $\epsilon \ll \delta{x}$ there
are insufficient number of points in the support of the kernel to estimate 
the surface geometry and the error becomes large.  We find when $\epsilon \gg \delta{x}$
is large there are many points within the support of the kernel, but the area 
of support is not localized enough to provide a good estimate of the operator.
We also show both the case with and without the area correction terms.
We find for our relatively uniform triangulations these give comparable 
overall errors here.  For our discretizations of the sphere based on $N=10,000$ points and $r_0 \geq 1.5$, we find 
that the optimal choice is $\epsilon \simeq 1$, see 
Figure~\ref{fig_conv_rel}. 
}

\subsection{Role of Geometry in First-Passage Times and Other Statistics}
\label{sec_u_stat}
We perform analysis to develop some results showing how 
our methods can be used for investigating the role of geometry 
of the first-passage times and other statistics associated
with the drift-diffusion dynamics of particles 
on curved surfaces.
Our Markov-Chain discretizations allow in some cases 
for computing efficiently statistics without the need to resort 
to Monte-Carlo sampling, provided the state space is not too large.
\begin{figure}
\centerline{\includegraphics[width=0.5\textwidth]{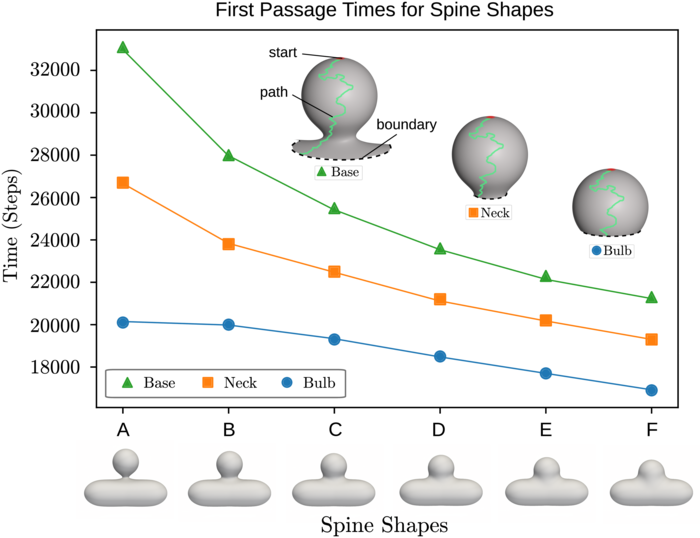}}
\caption{First-Passage Times for Spine Shapes.  First passage times are
computed when a protein starts at the top of the bulb-like head region and
reaches the boundary.  To investigate the role of the neck region verses other
aspects of the shapes, three cases are considered for the boundary location,
(i) at the base of the spine, (ii) in the middle of the neck, and (iii)
within the spherical head region just above the neck.  The results show that the
neck region plays the dominant role in the geometry.  As the neck narrows, the
diffusion of the protein to leave the head region and enter the tubular 
domain has a first-passage time that significantly increases from the geometry. 
} 
\label{fig_fpt} 
\end{figure}
We consider statistics of the form
\begin{equation} 
u_i^{(n)} = \Ex\left[f(X^{(N)}) +
\sum_{k=n}^{N - 1}g(X^{(k)},t) \; \middle|\; X^{(n)}=x_i \right].
\label{equ:stat_u_i_n}
\end{equation} 
Let $\mb{u}^{(k)}$ be the column vector with components
$[\mb{u}^{(k)}]_i = u_i^{(k)}$. 
The $f,g:\mathcal{M}\to\mathbb{R}$ are any two smooth functions
on the surface $\mathcal{M}$.  Let $[\mb{f}]_i = f(x_i)$ and 
$[\mb{g}^{(\ell)}]_i = g(x_i,\ell)$ be column vectors.
We also consider statistics of the form 
\begin{equation} 
w_i^{(\Omega)} = \Ex\left[f(X^{(\tau_\Omega)}) +
\sum_{k=0}^{\tau_\Omega - 1}g(X^{(k)}) \; \middle|\; X^{(0)}=x_i \right].
\label{equ:stat_u_omega}
\end{equation} 
For the domain $\Omega$, the 
$\tau_\Omega = \inf \{k \geq  0 \; | \; X^{(k)} \not\in \Omega\}$ is the 
stopping time index for the process to reach the boundary $\partial \Omega$.
Each of these statistics can be computed without sampling 
by the following results.
\begin{thm}
The statistics $\mb{u}^{(k)}$ of equation~\ref{equ:stat_u_i_n} satisfies
\begin{align}
\mathbf{u}^{(n-1)}&=M\mathbf{u}^{(n)}+\mathbf{g}^{(n-1)}\\
\mathbf{u}^{(N)}&=\mathbf{f}.  
\end{align} 
\end{thm}

\begin{proof} 
(see Appendix~\ref{sec_fpt})
\end{proof}

\begin{thm}
\rev{The statistics $\mb{w}^{(\Omega)}$ in equation~\ref{equ:stat_u_omega} 
satisfies}
\begin{eqnarray}
(\hat{M} - \hat{I}) \mb{w} &=& -\mb{g} \\
\partial{\mb{w}} &=& \mb{f}.
\end{eqnarray}
The $\partial{\mb{w}}$ extracts entries for all indices 
with $x_i \in \partial \Omega$.  The $\hat{M},\hat{I}$ 
refers to the matrix only with the rows with indices 
in the interior of $\Omega$.
\end{thm}

\begin{proof}
(see Appendix~\ref{sec_fpt})
\end{proof}
In the case that $\mb{f}=0,\mb{g} = \bsy{1}$, this becomes the First-Passage Time (FPT)
statistic $w_i=\Ex[\tau_\Omega\;|\; X^{(0)}=x_i]$.  These results allow for the 
statistics of equation~\ref{equ:stat_u_i_n} and~\ref{equ:stat_u_omega} 
to be computed efficiently without the need for Monte-Carlo sampling provided
the state space is not too large.  

\rev{Motivated by observations that the neck geometry appears to be a strong
factor in compartmentalization in dendritic
spines~\cite{Toennesen2014,Hugel2009}, we compute the first passage times of
protein diffusions for different spine shapes in Figure~\ref{fig_fpt}.  The
geometries were generated from isosurfaces in a technique commonly referred to
as meatballs~\cite{Blinn1982}.  We used several spheres to form the tubular
region and another sphere for the bulb region.  The geometry was obtained as
the isosurface using the level set function of the form $f(x) = \sum_{k}
\phi_k(|x - x_k|)$, where $\phi(r) = c_k/r^2$.  The surface was triangulated at
a refined spatial resolution to obtain around $N=29,000$ nodal points for each
shape.}

\rev{First passage times are computed for when a
protein starts at the top of the head region with the bulb-like shape and
reaches the boundary.  To investigate the role played by the neck region
compared to the influence of the other aspects of the geometry, three cases are
considered for the boundary location.  These are boundaries located (i) at the
base of the spine, (ii) in the middle of the neck, and (iii) within the
spherical head just above the neck.  These results indicate that compared to
the other geometric features, the width of the neck region plays the dominant
role in the first passage times.  As the neck narrows, the first-passage time
for the protein diffusion significantly increases from these changes in the
geometry, see Figure~\ref{fig_fpt}.
}

\begin{figure}
\centering
\includegraphics[width=0.45\textwidth]{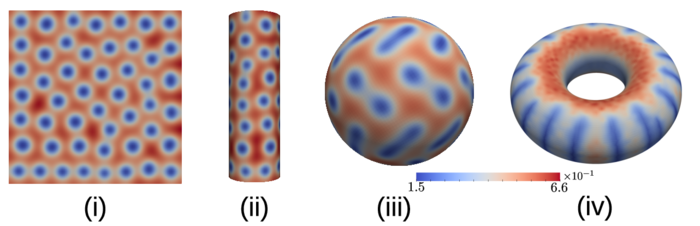}
\caption{Role of Topology and Geometry on Pattern Formulation.  The Gray-Scott
reaction-diffusion shows different patterns depending on the shape of the 
surface.  Shown are the cases of (i) square, (ii) cylinder, (iii) sphere, 
and (iv) torus. 
The surfaces have area one and when there are edges we use reflecting 
boundary conditions.  For the shapes (i)-(ii) spotted patterns emerge
having roughly a hexagonal pattern.  For spherical topology (iii) a 
regular hexagonal pattern without defects is no longer possible,
and instead striped patterns mix with spots.  For the case of
a torus (iv), which can sustain a hexagonal pattern in principle, 
the heterogeneity of the curvature appears to drive the formation of 
localized stripe-like patterns.  These results indicate that both 
the geometry and topology can significantly impact pattern formation.
} 
\label{fig_gray_scott_geo}
\end{figure}

\begin{table}
    \centering
    \begin{tabular}{l | l | l}
        Parameter &     Description &                           Value \\
        \hline\hline
        $\epsilon_u$ &  diffusion scale for $u$ &            $1.0\times 10^{-2}$\\
        $\epsilon_v$ &  diffusion scale for $v$ &            $5.0\times 10^{-3}$\\
        $r_0$ &           diffusion cut-off radius &               1.0\\
        $a$ &           $u+2v\rightarrow 3v$ reaction rate &    $4.0\times 10^{-2}$\\
        $b$ &           $v\rightarrow p$ reaction rate &        $6.0\times 10^{-2}$\\
        $T_\text{sim}$ &           total time duration &                            $5.0\times 10^3$\\
        $\Delta t$ &    time-step &                         $1.0\times 10^{-1}$\\
        $n_{rk}$ &      Runge-Kutta steps &                     100
    \end{tabular}
    \caption{Parameters for Gray-Scott Reaction-Diffusion 
    System in Figure~\ref{fig_gray_scott_geo}.}
    \label{tab:GS_shapes}
\end{table}

\begin{figure} 
\centering 
\includegraphics[width=0.45\textwidth]{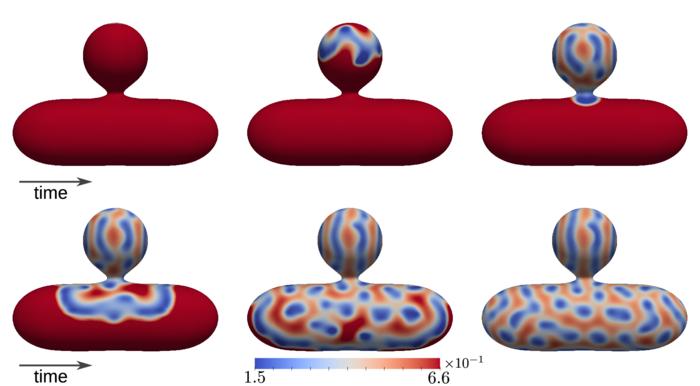}
\caption{\rev{Evolution of Turing Instabilities for Gray-Scott Reactions.  Shown
is the evolution of the reaction-diffusion pattern formation process on
the dendritic spine shape having the narrowest neck, label $A$ (notation the same as in Figure~\ref{fig_fpt}).  The pattern
progresses first by forming within the bulb-like head region, and then 
spreads through the neck to form patterning on the tubular part of 
the domain.  The time-steps are shown 
for $n=0,400,800,1200,1600,2000.$}}
\label{fig_gray_scott_evolve} 
\end{figure}

\section{Simulations} 
\label{sec_simulations}

\subsection{Turing Instabilities on Curved Surfaces} 
\label{sec_rd}

We show how our discretization approach can be used to develop
methods for performing simulations of general reaction-diffusion 
processes on curved surfaces of different shapes.  
We consider reaction-diffusions, 
such as the pattern formation process based
on Turing's instability mechanism~\cite{Turing1952},
where the geometry and topology of the domain can impact 
the patterns that are obtained~\cite{Britton2005,Murray2013}.
Consider the system with two molecular species 
with concentrations $u,v$ 
\begin{equation} 
\pdv{u}{t} =D_u\Delta_{\mathcal{M}}u + f(u,v),
\hspace{0.3cm} \pdv{v}{t} =D_v\Delta_{\mathcal{M}}v + g(u,v).
\end{equation}
The diffusivities $D_u$ and $D_v$ will in general
be different. Through the non-linear reaction terms,
the difference in diffusivity can cause the 
homogeneously mixed concentrations to become 
unstable resulting in pattern generation~\cite{Turing1952,Britton2005}.

We consider Gray-Scott reactions~\cite{Gray1984}, which can exhibit 
different patterns depending on the initial conditions and
interactions with noise and other 
perturbations~\cite{Atzberger_RD_2010,Britton2005}.
For the Gray-Scott reactions~\cite{Gray1984}, the terms 
are $f(u,v) = -uv^2 + a(1 - u)$
and $g(u,v) = uv^2 - (a + b)v$.  The 
rate parameters $a,b$ are
for the chemical reactions $u + 2v\xrightarrow[]{a} 3v$ and 
$v \xrightarrow[]{b} p$.

\begin{table}
    \centering
    \begin{tabular}{l | l | l}
        Parameter &     Description &                           Value \\
        \hline\hline
        $\epsilon_u$ & diffusion scale for $u$  &            $1.0\times 10^{-1}$\\
        $\epsilon_v$ & diffusion scale for $v$  &            $5.0\times 10^{-2}$\\
        $r_0$ & diffusion cut-off radius &               1.0\\
        $a$ & $u+2v\rightarrow 3v$ reaction rate &    $4.0\times 10^{-2}$\\
        $b$ & $v\rightarrow p$ reaction rate &        $6.0\times 10^{-2}$\\
        $T_\text{sim}$ & total time duration &                            2.0\\
        $\Delta t$ &  time-step &                         $1.0\times 10^{-3}$\\
        $n_{rk}$ &    Runge-Kutta steps &                     1
    \end{tabular}
    \caption{Parameters for Gray-Scott Reaction-Diffusion 
    System in Figures~\ref{fig_gray_scott_evolve} and ~\ref{fig_gray_scott_spine}.} 
    \label{tab:GS_DS}
\end{table}
\begin{figure} \centering
\includegraphics[width=0.499\textwidth]{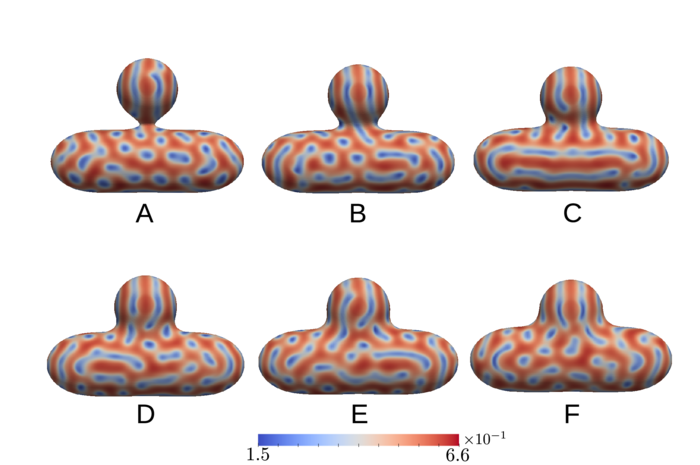}
\caption{Gray-Scott Reaction-Diffusion Pattern Formation on Dendritic Spine
Neck Shapes.  \rev{Starting with the steady-state $(u^*,v^*)=(1,0)$, given this
  is stationary and the evolution is deterministic, we apply a perturbation in
  the shape of a circular patch on the top of the head region with the
  bulb-like shape with $(u,v) = (0.5\pm10\%,0.25\pm10\%)$.  The perturbations
  are independent uniform variates at each lattice site. The reactions
  settle down into the patterns shown for time-step $2000$. Striped patterns
  manifest near the base of the head region and mixed spotted patterns within
  the tubular region.}}
\label{fig_gray_scott_spine} 
\end{figure}

We start with a homogeneous steady-state solution for the system 
$(u,v) = (1,0)$, and add small perturbations based on 
uniform random noise $\pm10\%$ of the steady-state 
at each location $x_i$ within a region $\Gamma \subseteq \Omega$. 
Throughout our simulations, we choose $a = 0.04$, $b=0.06$,
guided by the parameters of the phase diagram of~\cite{Pearson1993}.
We investigate how the Gray-Scott patterns are influenced by
different geometries and topologies.

\rev{Using our Markov-Chain discretizations, we develop numerical methods for 
the spatial-temporal evolution of the concentration fields.
By equation \eqref{diffprob}, we obtain a stochastic matrix $M$. 
The overall idea is to model 
the concentrations by $u(x,t) = (n_u/A) p_u(x,t)$,
$v(x,t) = (n_v/A) p_v(x,t)$, where $A$ is the surface area.
The $\mb{p} = [\mb{p}_u,\mb{p}_v]$, with $[\mb{p}^n]_i = p(x_i,t_n)$, 
are obtained from the probability evolution for the Markov-Chain
given by $\mb{p}^n = \mb{p}^{n+1}M $.  In this way, we obtain 
a model that approximates the diffusive evolution of the continuous 
concentration fields.}  

\rev{To model the full reaction-diffusion system 
in our simulations, we split
the time-step integration into a diffusive step and a reaction step. 
For the diffusive step $\Delta{t}$, we use our Markov-Chain 
discretization in Section~\ref{sec_detailed_balance} 
to update the concentration fields $q=(u,v)$.
This provides a mean-field model for the concentration field, with 
$q^{n + 1/2} = \Psi_1(q^n)$ corresponding to $u^{n+1/2} = u^n M, v^{n+1/2} = v^n M$.  
The $M$ is the transition matrix for the surface diffusion
derived in Section~\ref{sec_detailed_balance}.
For the reactions,
we use the smaller time-steps $\delta{t} = \Delta{t}/n_{rk}$ which are 
integrated with the fourth-order
Runge-Kutta method RK4~\cite{Burden2010,Iserles2008}.  This gives 
$q^{n + 1} = \Psi_2(q^{n + 1/2})$.
We alternate between these steps in our simulations to obtain 
$q^{n+1} = \Psi(q^n) = \Psi_2(\Psi_1(q^n))$.}  

\rev{ As we discussed in Section~\ref{sec_markov_chain}, the $M$ approximates
the Laplace-Beltrami operator.  This ensures with enough spatial resolution the
evolution will approximate diffusion on a surface.  One way to view our model
is as a meshfree way to discretize the surface reaction-diffusion PDEs.  This
provides simulation methods corresponding to the mean-field concentration
fields of the protein species.  }

\rev{Our approaches for the diffusion 
also can be used more generally to go beyond the mean-field model by
using discrete particle simulations with individual 
random walkers.  These would have
the same distribution as our continuum model when 
taking appropriate limits.  Our methods allow for 
either (i) to perform stochastic 
simulations of random walks tracking individual particles 
to account for discrete spatial-temporal fluctuations from 
finite number effects, or (ii) to perform deterministic 
simulations tracking the probability distribution of the
walkers.  For the reaction-diffusion studies we use approach 
(ii).  For later studies involving phase separation,
we use approach (i).}

We consider the Gray-Scott reaction-diffusion process on the following
geometries (i) flat sheet with boundaries, (ii) finite cylinder with
boundaries, (iii) surface of a sphere, (iv) surface of a torus, 
see Figure~\ref{fig_gray_scott_geo}.  The surfaces have area one and when 
there are boundary edges we use reflecting boundary conditions.  For the shapes 
(i)-(ii) spotted patterns emerge having roughly a hexagonal pattern.  
For the spherical topology (iii) a regular hexagonal pattern 
without defects is no longer possible, and instead striped 
patterns mix with spots.  For the case of a torus (iv), 
which can sustain a hexagonal pattern in principle, 
the heterogeneity of the curvature appears to drive the formation of 
localized stripe-like patterns.  The results obtained with our
methods indicate that both the geometry 
(curvature and scale effects) and the topology can
impact significantly the pattern formation process.

\rev{We remark that Turing instabilities as a mechanism for pattern formation
were originally motivated by a linear stability analysis on periodic domain
performed by A. M. Turing in~\cite{Turing1952}.  Extending the analysis to
general manifolds is non-trivial given the challenges of analyzing even for the
linearized dynamics the eigenfunctions and eigenvalues associated with the
Laplace-Beltrami based-diffusions and roles of curvature and topology.  Some
work in this direction has been done in the literature in the case of
specialized geometry, such as the sphere~\cite{Varea1999,JamiesonLane2016}.
The results of our methods for Turing instabilities yield patterns comparable
to those seen in such previous studies.  In the spherical case we find that
there is a transition from striped to spotted patterns in the sphere
size~\cite{Hinz2019}.  We also find that the hexagonal spotted patterns with
defects manifest with arrangements similar
to~\cite{Varea1999,JamiesonLane2016}.  }

We also consider more complicated shapes for the Gray-Scott reactions 
given by our mechanistic dendritic spine geometries with different 
neck sizes, see Figures~\ref{fig_gray_scott_evolve}, \ref{fig_gray_scott_spine}.  
These shapes consist of a bulb-like head region (representing the spine)
which is connected by a neck structure to a tube-like region 
(representing the dendrite).  We vary the shapes by changing the 
thickness of the neck-like structure joining the two regions. 
The evolution of 
the pattern formation process for the geometry with the narrowest 
and widest neck shapes are shown in 
Figure~\ref{fig_gray_scott_evolve}.
For the narrowest necks, the confinement in
the head region appears to result in stripe-like patterns that 
also extend through the neck.  In the larger tube-like region, 
the spotted patterns form inter-mixed with the stripe pattern, 
see Figure~\ref{fig_gray_scott_spine}.  These results indicate how
local regions can exhibit different patterning depending on the 
local geometry.

\subsection{Dendritic Spines and Protein Kinetics} 
\label{sec_phase_sep}
\label{sec_spine}

\begin{figure} 
\centering 
\includegraphics[width = 0.45\textwidth]{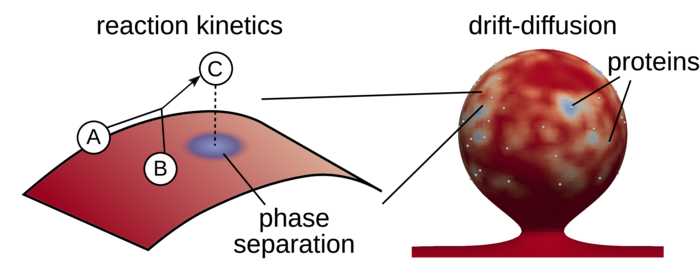} 
\caption{Chemical Kinetics and Phase Separation.  We consider 
chemical kinetics where molecular species $A,B$ diffuse freely
and react to form $C$ by $A + B \rightarrow C$.  The molecular 
species $C$ can nucleate phase $q=-1$.  The drift-diffusion
of $C$ is coupled to the local phase by equation \ref{protein_phase_int}.  This
results in a bi-directional confinement force from the phase 
field $q$ acting on $C$ to keep within regions with $q=-1$.  
This also results in an equal-and-opposite force acting on 
the phase field from equation \ref{protein_phase_int} driving nucleation
and phase separation.
} \label{fig_model_spine}
\end{figure}

We develop a mechanistic model for dendritic spines to 
investigate the dependence of protein transport and kinetics
on geometry and heterogeneities arising from
phase separation.
Our investigations are motivated from experiments on SynGAP
and PSD-95 proteins, where phase separation may play a role in driving receptor
organization~\cite{Zeng2016}.  Phase separation can arise 
from nucleation or modulating local concentrations in the cell membrane,
which is a heterogeneous mixture of lipids, 
proteins, and other small molecules~\cite{Simons2004,Lingwood2010,Alberts2014}. 

In our mechanistic model, we investigate how the spine geometry and phase
separation can influence reaction kinetics. We start with a two species system.
The $A$ and $B$ molecular species are tracked at the individual particle level.
In the absence of phase separation, this would diffuse freely over the membrane
surface.  As a starting point, we study the basic chemical kinetics $A + B
\rightarrow C$.  The discrete particles react with probability $p$ when
coming within a reaction distance $r < r_0$.  This is motivated by Smoluchowski
reaction kinetics~\cite{Smoluchowski1918,Gillespie2013,
Isaacson2013,IsaacsonPeskin2011}.  

\begin{figure} 
\centering
\includegraphics[width = 0.48\textwidth]{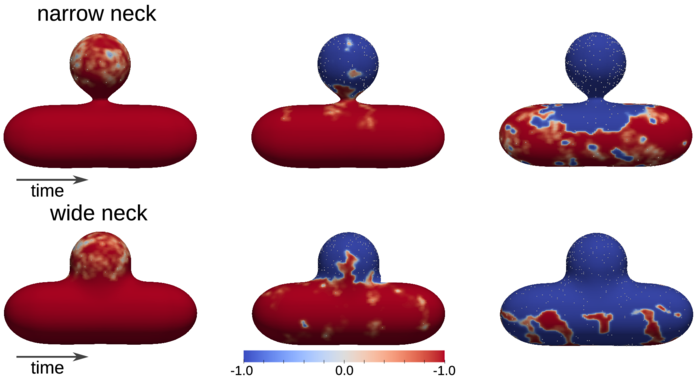} 
\caption{\rev{Dendritic Spine Model: Role of Geometry in Phase Separation.  The
dendritic spine model consists of a head region connected by a neck
region to a tubular domain.  The active proteins $A,B$ in the model all
originate at the top of the head region.  The phase separation occurring
in the two cases is shown for (i) a narrow neck constricting the protein
diffusion and spread of the phase separation (top), (ii) a wide neck through
which the proteins can readily diffuse and the phase can separate (bottom).
The phase $q=1$ is red and $q=-1$ is blue, shown are time-steps $50$, $300$, and
$2000$.  At around time-step $\sim 300$ (middle), the the protein diffusion and
phase separation comes into contact with the neck region.  In the narrow
case, the neck region acts effectively like an entropic barrier for the
diffusion arising from the geometry.}}
\label{fig_phase_evolve} 
\end{figure}

\begin{table}[]
    \centering
    \begin{tabular}{l|l|l}
         Parameter & Description & Value \\
         \hline\hline
         $\alpha_1$ &       GL phase interfacial tension &      $1.0\times 10^6$\\
         $\alpha_2$ &       GL phase polarity  &       $1.0\times 10^4$\\
         $\alpha_3$ &       GL protein-phase coupling &       $1.0$\\
         $a$        &       protein phase radius & $0.1$\\
         $N$        &       $A$, $B$ initial protein count      &       1000\\
         $\epsilon$ &       diffusion scale for proteins &    $0.1$\\
         $r_0$        &       diffusion cut-off radius &   $0.5$\\
         $p$        &       $A + B\to C$ reaction probability & $0.01$\\
         $T_\text{sim}$        &       total time duration              & $200.0$\\
         $\Delta t$ &       time-step           & $0.1$\\
         $n_{rk}$ &                  Runge-Kutta steps &         $100$
    \end{tabular}
    \caption{Parameters for Dendritic Spine Phase Separation.}
    \label{tab:phase_sep}
\end{table}

To model such effects at a coarse-level, we track a continuum phase 
field $\phi = \phi(x,t)$ which can 
couple to the local diffusive motions of the discrete protein particles.  We consider in 
our initial model the case of a local order parameter associated with 
phase separation based on Ginzburg-Landau (GL) 
theory~\cite{Hohenberg2015,Ginzburg1950}. 
Other phase-separation phenomena and approaches also could be 
considered in principle within our framework, such as the second-order 
Allen-Cahn~\cite{AllenCahn1979} or the conservative fourth-order 
Cahn-Hilliard~\cite{Cahn1958}, with additional 
work on the numerical methods for the operators on curved 
surfaces and coupling with the Markov-Chain 
discretization~\cite{Atzberger_GMLS_Hydrodynamics_2020}.
For notational convenience, we will denote 
$q \simeq \phi(x)$, so $q: \mathcal{M} \rightarrow \mathbb{R}$
for the phase variable map for the curved surface $\mathcal{M}$.
This gives at lattice site $i$,
$q_i \simeq \phi(x_i)$. 
The GL functional is $\hat{V}_0[\phi] = \int_{\mathcal{M}} \left(\nabla \phi(x)\right)^2 
+ \hat{V}_2[\phi](x) dx$. The first term accounts for a line tension between phases 
and the second term drives the phase to $\pm 1$ with local energy density 
$\hat{V}_2[\phi] = K (1 - \phi^2)^2$.
To capture similar effects as $\hat{V}_0[q]$, we use the simplified discrete model
\begin{equation} 
  V_0[q]=\frac{1}{n^2}\sum_{i}\sum_{j\neq i}
V_{1}[q_{i},q_{j}] + \frac{1}{n}\sum_{i}V_{2}[q_{i}],
\label{V}
\end{equation} 
with
\begin{equation} 
V_{1}[q_{i},q_{j}]=\alpha_{1}
  W_{ij}(q_{i}-q_{j})^2, \hspace{0.3cm}
V_{2}[q]=\alpha_{2} (1-q^{2})^{2}.
  \label{V1_V2}
\end{equation} 
The first term models the interfacial tension, where the
coefficient
$W_{ij} = W(r_{ij})$ decays in $r_{ij} = |x_i - x_j|$.
The second term models the order parameter which locally
is at a minimum for $q = \pm 1$.  The $\alpha_{1}, \alpha_2 > 0$ 
control the strength of the interfacial tension and 
the phase variable ordering.  We use for the decay coefficient
$W_{ij}=\exp(-r_{ij}^{2}/\varepsilon^{2})$
which is truncated for large $r_{ij} \gg \epsilon$.

To couple protein motions $\mb{X}(t)$ and their ability to induce 
local phase ordering in $q$, we use the free energy term
\begin{equation}
  V_3^{(i)}[q,X]=\alpha_3(q_i- (-1))^2\eta_a(x_i-X).\label{protein_phase_int}
\end{equation}
The $\alpha_3>0$ controls the strength of this coupling.
The $\eta_a$ gives a kernel function with support 
localized around the protein location $X = X(t)$.
This term drives the phase field toward $q=-1$ in a 
region around the protein location $X(t)$.
We use $\eta_a(r) = \exp(-r^2/a^2)$ and 
\begin{equation}
  V_3[q,X]= \frac{1}{n} \sum_i V_{3}^{(i)}[q,X].
\end{equation}
The free energy for full system $(q,X)$ with protein configuration $X$ 
and phase-field $q$ is  
\begin{equation}
V[q,X] = V_0[q] + \sum_k V_3[q,X^k].
\label{full_eng}
\end{equation}
The discrete protein positions $[\mb{X}]_k = X^k$ are updated using 
the Markov-Chain discretization 
in equation \eqref{equ:M} for the drift-diffusion dynamics.
The energy for the protein configuration is given 
by $U(\mb{X}) = V[q,\mb{X}]$.
The time evolution of the phase field is given by
\begin{equation} 
\frac{d q_{i}}{dt}=-\nabla_{q_i}V[q,X].
\label{phase_evolve}
\end{equation}
This is discretized and integrated over sub-time-steps 
$\delta{t} = \Delta{t}/n_{rk}$
using Runge-Kutta RK4~\cite{Burden2010,Iserles2008}.  
Modeling the discrete system using the common free energy $V$
ensures the bi-directional coupling gives 
equal-and-opposite forces between the phase-field $q$ and the 
proteins $X$.  

\begin{figure}
\includegraphics[width=0.4\textwidth]{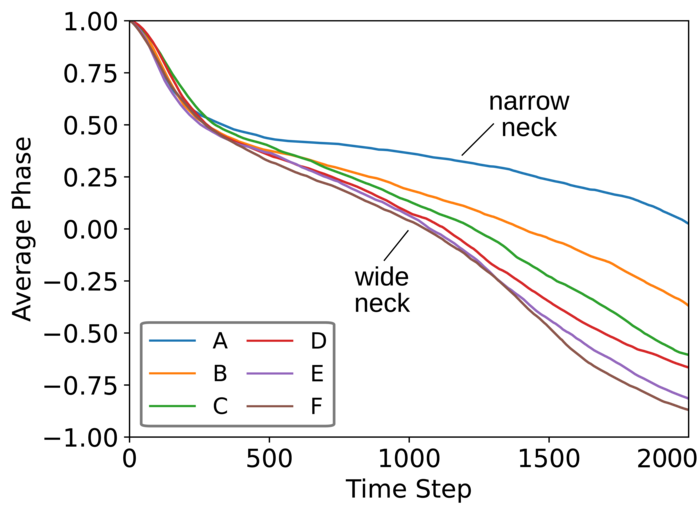}
\caption{Dendritic Spine Model: Average Phase.  The proteins $C$
can nucleate phase with $q=-1$.  Shown are how the geometry impacts
phase evolution for five different sizes for the neck region with
labels A-F ordered from narrowest to widest as in Figure \ref{fig_gray_scott_spine}.  
The phase separation is found to slow down considerably for the 
narrowest neck (label A) relative to the widest neck (label F).
The differences emerge around around time-step $\sim 300$, when 
the bulb-like head region is saturated.  We see over the 2000 time-steps
the widest neck results in almost the entire domain including the
tubular region converting to phase $q=-1$.  For the narrowest neck
shape, we see at time-step $2000$ the phase separation is primarily 
isolated to the head head region and surrounding area, while the 
rest of the domain remains primarily $q=1$.  \rev{The average is taken over the entire surface.}  }
\label{fig_avg_phase} 
\end{figure}

We perform simulations to investigate how the dendritic spine geometry 
impacts the protein reaction kinetics and phase-separation,
see Figure \ref{fig_model_spine} and Table~\ref{tab:GS_shapes}.
The proteins couple to the phase separation 
by having the ability to drive nucleation of local patches with
$q=-1$ near the protein location $X(t)$.  The geometries were chosen 
with the neck region taking on shapes varying from narrow to wide,
see Figure \ref{fig_gray_scott_geo}.  The protein species $A,B$ are modeled 
as originating in the top head region of the spine.  This could 
arise for instance if these species correspond to tracking proteins 
only after they have 
become activated in this region.  The proteins can then diffuse,
and may induce local phase-separation through the coupling $V_3$.
The proteins can also interact to form a complex which is tracked by
molecular species $C$.  

We show how the phase 
separation proceeds over time as the size of the neck region 
varies in Figure \ref{fig_phase_evolve}.  As the neck 
region narrows, the protein diffusion and the phase separation 
are restricted to be localized in the head region.  
As the neck becomes wide, the protein diffusion and phase separation
can readily proceed to spread more rapidly into the tubular region,
see Figures~\ref{fig_phase_evolve}, \ref{fig_avg_phase}.  Further
effects arise from the coupling of the proteins with the local phase.

\begin{figure}
\includegraphics[width=0.4\textwidth]{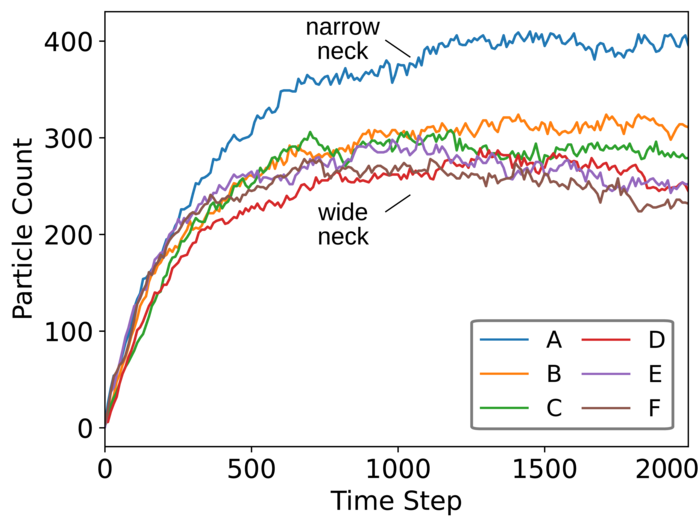}
\caption{Dendritic Spine Model: Proteins $C$ in the Head Region.  
Shown are the number of proteins $C$ within the bulb-like head region over time.
We see the reactions producing $C$ proceed almost independent of the geometry
up to time-step $\sim 300$.  Afterwards, the shapes with the narrowest necks
results in far more $C$ proteins being produced and retained within the 
head region.  This shows how phase-separation can serve to enhance 
retaining proteins near the top of the spine.
} \label{fig_C_count} 
\end{figure}

The protein dynamics are impacted by both the geometry
and the local phase.  In our simulation studies, the 
coupling coefficients are taken to 
be $\alpha_3^A=\alpha_3^B = 0$ and $\alpha_3^C = 1.0$.
For this case $C$ can nucleate phase separation nearby since it 
prefers the phase field have $q=-1$.  
This also results for $C$ experiencing a trapping force within regions 
with $q=-1$, from the free energy in equation \eqref{protein_phase_int}.
This restricts the movements of $C$.  The evolution of the creation of 
of molecular species $C$ is shown in Figure~\ref{fig_C_count}.  
The number of $C$ formed and retained in the head region is significantly
impacted by the neck geometry and phase separation.  When the neck is 
narrow, the confinement of $A,B$ and the phase separation to the head 
region, both enhances the creation of $C$ by more frequent $A-B$ encounters
and in $C$'s retention to the head region from the phase trapping forces.
As the neck becomes wide, the $A,B$ can diffuse more freely and when
the phase does nucleate it can more rapidly spread throughout 
the whole domain.  This results in a much smaller number of 
$C$ being created and retained in the head region.

Our basic mechanistic model shows that the morphology 
and phase separation can interact to serve together 
to enhance retaining proteins near the top of the dendritic spine.
These results are expected to carry over to more complicated 
chemical kinetic systems.  The further interactions between
diffusion, kinetics, phase separation, and the geometry 
can regulate in different ways the spatial arrangements and the 
local effective
rates of reactions in curved cell membranes.

\subsection{Protein Clustering and Role of Phase Separation}
\label{sec_clustering_phase_sep}

\begin{table}[]
    \centering
    \rev{
    \begin{tabular}{l|l|l}
         Parameter & Description & Value \\
         \hline\hline
         $\alpha_1$ &       GL phase interfacial tension &      $1.0\times 10^6$\\
         $\alpha_2$ &       GL phase polarity  &       $1.0\times 10^4$\\
         $\alpha_3$ &       GL protein-phase coupling &       0 -- 10.0\\
         $a$        &       protein phase radius & $0.1$\\
         $N$        &       $A$, $B$ initial protein count      &       1000\\
         $\epsilon$ &       diffusion scale for proteins &    $0.1$\\
         $r_0$        &       diffusion cut-off radius &   $0.5$\\
         $p$        &       $A + A\to 2A$ reaction probability & $0.001$\\
         $q$        &       $A \to B$ decay probability & $0.001$\\
         $T_\text{sim}$        &       total time duration              & $200.0$\\
         $\Delta t$ &       time-step           & $0.1$\\
         $n_{rk}$ &                  Runge-Kutta steps &         $100$
    \end{tabular}
    }
    \caption{\rev{Parameters for the dendritic spine phase separation simulations.}}
    \label{tab:ABAA2A}
\end{table}

\rev{ We investigate how phase separation can impact protein kinetics and
clustering.  Consider the case of states where proteins in state $A$ are in an
active state receptive to binding and proteins in state $B$ are dormant.
Consider the competing reactions for forming dimers}
\rev{
\begin{equation}
    A + A\xrightarrow{p} 2A,\qquad A \xrightarrow{q} B.
\end{equation}
}
\rev{The $p$ and $q$ denote here the reaction probabilities.  The
protein-phase-field interactions are modeled using our approach in Section
\ref{sec_phase_sep}.  We take the phase field to impact the proteins in states
$A$ and $2A$ with the same strength, $\alpha_3^A=\alpha_3^{2A}$.  In contrast,
we take proteins in the dormant state $B$ to not interact with the phase field,
$\alpha_3^B=0$.  To study the effects of the phase field, we vary $\alpha_3^A$
over the range from $0.0$ (no interaction) to $1\times 10^1$ (strong
interaction).  We investigate in simulations the number of dimers $2A$ proteins
that are formed over time for different levels of phase field interaction
$\alpha_3^A$ and as the phase separation progresses.  The $2A$ counts are given
in Figure~\ref{fig_2A_count}.  We find in the strongest coupling cases the
phase field can significantly impact the reaction kinetics promoting clustering
and dimer formation.  The details of the
parameters used in our simulations are given in Table \ref{tab:ABAA2A}.}

\begin{figure}
\includegraphics[width=0.8\columnwidth]{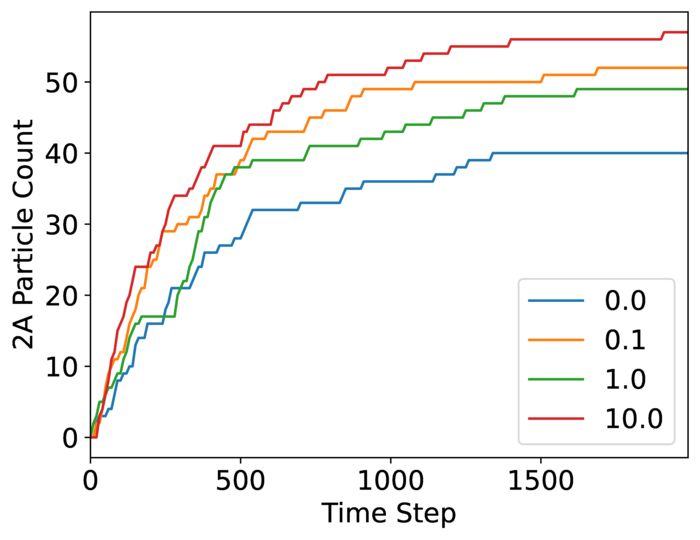}
\caption{Dendritic Spine Model: 
Shown are the number of proteins $2A$ within the bulb-like head region over
time.  We see that increasing the interaction strength $\alpha_3$ (legend)
between proteins $A$ and the phase field promotes the formation of 
clusters $2A$.
} \label{fig_2A_count} 
\end{figure}

\rev{The results show that phase-field separation can play significant roles
impacting reaction kinetics within membranes.  In this case promoting the
formation of clusters.  The underlying mechanisms has to do with the formation
for proteins of local environments with initially small phase-separated domains
that serve to confine the diffusion of the proteins.  Relative to free
diffusion over the entire membrane surface, the phase confinement results in
diffusion with the phase domain and more frequent encounters between
co-confined proteins.  In this way, the binding kinetics are enhances relative
to the free diffusion.  From the greater slope in the increase in the number of
clusters formed for case $\alpha_3^A=1\times 10^1$ vs $\alpha_3^A=0$, we see
the phases have the greatest impact when there are small phase domains.  The
smaller domains results in more frequent encounters between confined proteins.
As the phase separation progresses the domains become larger and merge with the
effect of the phase-coupling on the reactions become less pronounced.}

\rev{These results show how phase separation and domain formation can
be utilized within cell membranes to significantly augment the reaction kinetics
without changing the reaction characteristics of the individual proteins.  These simulation methods provide further ways to investigate the collective kinetics of proteins and the roles played by phase composition and the membrane geometry.}


\section{Conclusions}
We have developed biophysical models 
for investigating protein kinetics within
heterogeneous cell membranes of arbitrary
shape.  The introduced methods allow for 
studying at the continuum and 
discrete protein-level the drift-diffusion
dynamics and reaction kinetics taking into 
account the roles of geometry and coupling 
to membrane phase separation.  We have shown
that coupling of kinetics to 
phase separation and geometry can 
significantly enhance or inhibit 
reactions and cluster formation 
processes.  These initial results 
suggest a few mechanisms
by which biological membranes may 
utilize phase separation and geometry 
to drive kinetics and organization of 
proteins within cell membranes. \\

\section{Acknowledgements}
The authors P.J.A. and P.T. would like to acknowledge support for 
this research from the grants NSF Grant DMS-1616353 and 
DOE Grant ASCR PHILMS DE-SC0019246.  The author 
T.A.B. would like to 
acknowledge support from grant NIH R37MH080046.  P.T. 
would like to acknowledge a Goldwater Fellowship and 
College of Creative Studies Summer Undergraduate Fellowship.  
P.J.A. would like to acknowledge a hardware grant from Nvidia.
Authors also would like to acknowledge UCSB Center for 
Scientific Computing NSF MRSEC (DMR1720256) and 
UCSB MRL NSF CNS-1725797.  

\bibliography{paper_database}

\begin{thebibliography}{110}%
\makeatletter
\providecommand \@ifxundefined [1]{%
 \@ifx{#1\undefined}
}%
\providecommand \@ifnum [1]{%
 \ifnum #1\expandafter \@firstoftwo
 \else \expandafter \@secondoftwo
 \fi
}%
\providecommand \@ifx [1]{%
 \ifx #1\expandafter \@firstoftwo
 \else \expandafter \@secondoftwo
 \fi
}%
\providecommand \natexlab [1]{#1}%
\providecommand \enquote  [1]{``#1''}%
\providecommand \bibnamefont  [1]{#1}%
\providecommand \bibfnamefont [1]{#1}%
\providecommand \citenamefont [1]{#1}%
\providecommand \href@noop [0]{\@secondoftwo}%
\providecommand \href [0]{\begingroup \@sanitize@url \@href}%
\providecommand \@href[1]{\@@startlink{#1}\@@href}%
\providecommand \@@href[1]{\endgroup#1\@@endlink}%
\providecommand \@sanitize@url [0]{\catcode `\\12\catcode `\$12\catcode
  `\&12\catcode `\#12\catcode `\^12\catcode `\_12\catcode `\%12\relax}%
\providecommand \@@startlink[1]{}%
\providecommand \@@endlink[0]{}%
\providecommand \url  [0]{\begingroup\@sanitize@url \@url }%
\providecommand \@url [1]{\endgroup\@href {#1}{\urlprefix }}%
\providecommand \urlprefix  [0]{URL }%
\providecommand \Eprint [0]{\href }%
\providecommand \doibase [0]{https://doi.org/}%
\providecommand \selectlanguage [0]{\@gobble}%
\providecommand \bibinfo  [0]{\@secondoftwo}%
\providecommand \bibfield  [0]{\@secondoftwo}%
\providecommand \translation [1]{[#1]}%
\providecommand \BibitemOpen [0]{}%
\providecommand \bibitemStop [0]{}%
\providecommand \bibitemNoStop [0]{.\EOS\space}%
\providecommand \EOS [0]{\spacefactor3000\relax}%
\providecommand \BibitemShut  [1]{\csname bibitem#1\endcsname}%
\let\auto@bib@innerbib\@empty
\bibitem [{\citenamefont {Alberts}\ \emph {et~al.}(2014)\citenamefont
  {Alberts}, \citenamefont {Johnson}, \citenamefont {Lewis}, \citenamefont
  {Morgan}, \citenamefont {Raff}, \citenamefont {Roberts}, \citenamefont
  {Walter}, \citenamefont {Wilson},\ and\ \citenamefont {Hunt}}]{Alberts2014}%
  \BibitemOpen
  \bibfield  {author} {\bibinfo {author} {\bibfnamefont {B.}~\bibnamefont
  {Alberts}}, \bibinfo {author} {\bibfnamefont {A.}~\bibnamefont {Johnson}},
  \bibinfo {author} {\bibfnamefont {J.}~\bibnamefont {Lewis}}, \bibinfo
  {author} {\bibfnamefont {D.}~\bibnamefont {Morgan}}, \bibinfo {author}
  {\bibfnamefont {M.~C.}\ \bibnamefont {Raff}}, \bibinfo {author}
  {\bibfnamefont {K.}~\bibnamefont {Roberts}}, \bibinfo {author} {\bibfnamefont
  {P.}~\bibnamefont {Walter}}, \bibinfo {author} {\bibfnamefont {J.~H.}\
  \bibnamefont {Wilson}},\ and\ \bibinfo {author} {\bibfnamefont
  {T.}~\bibnamefont {Hunt}},\ }\href@noop {} {\emph {\bibinfo {title}
  {Molecular biology of the cell, 6th Edition}}}\ (\bibinfo  {publisher}
  {Garland Science},\ \bibinfo {year} {2014})\BibitemShut {NoStop}%
\bibitem [{\citenamefont {Roux}\ \emph {et~al.}(2005)\citenamefont {Roux},
  \citenamefont {Cuvelier}, \citenamefont {Nassoy}, \citenamefont {Prost},
  \citenamefont {Bassereau},\ and\ \citenamefont {Goud}}]{BassereauRoux2005}%
  \BibitemOpen
  \bibfield  {author} {\bibinfo {author} {\bibfnamefont {A.}~\bibnamefont
  {Roux}}, \bibinfo {author} {\bibfnamefont {D.}~\bibnamefont {Cuvelier}},
  \bibinfo {author} {\bibfnamefont {P.}~\bibnamefont {Nassoy}}, \bibinfo
  {author} {\bibfnamefont {J.}~\bibnamefont {Prost}}, \bibinfo {author}
  {\bibfnamefont {P.}~\bibnamefont {Bassereau}},\ and\ \bibinfo {author}
  {\bibfnamefont {B.}~\bibnamefont {Goud}},\ }\href
  {http://www.ncbi.nlm.nih.gov/pmc/articles/PMC1142567/} {\bibfield  {journal}
  {\bibinfo  {journal} {The EMBO Journal}\ }\textbf {\bibinfo {volume} {24}},\
  \bibinfo {pages} {1537} (\bibinfo {year} {2005})}\BibitemShut {NoStop}%
\bibitem [{\citenamefont {Noske}\ \emph {et~al.}(2008)\citenamefont {Noske},
  \citenamefont {Costin}, \citenamefont {Morgan},\ and\ \citenamefont
  {Marsh}}]{Noske2008}%
  \BibitemOpen
  \bibfield  {author} {\bibinfo {author} {\bibfnamefont {A.~B.}\ \bibnamefont
  {Noske}}, \bibinfo {author} {\bibfnamefont {A.~J.}\ \bibnamefont {Costin}},
  \bibinfo {author} {\bibfnamefont {G.~P.}\ \bibnamefont {Morgan}},\ and\
  \bibinfo {author} {\bibfnamefont {B.~J.}\ \bibnamefont {Marsh}},\ }\bibfield
  {booktitle} {\emph {\bibinfo {booktitle} {The 4th International Conference on
  Electron TomographyThe 4th International Conference on Electron
  Tomography}},\ }\href
  {http://www.sciencedirect.com/science/article/pii/S1047847707002316}
  {\bibfield  {journal} {\bibinfo  {journal} {Journal of Structural Biology}\
  }\textbf {\bibinfo {volume} {161}},\ \bibinfo {pages} {298} (\bibinfo {year}
  {2008})}\BibitemShut {NoStop}%
\bibitem [{\citenamefont {Adler}\ \emph {et~al.}(2010)\citenamefont {Adler},
  \citenamefont {Shevchuk}, \citenamefont {Novak}, \citenamefont {Korchev},\
  and\ \citenamefont {Parmryd}}]{Adler2010}%
  \BibitemOpen
  \bibfield  {author} {\bibinfo {author} {\bibfnamefont {J.}~\bibnamefont
  {Adler}}, \bibinfo {author} {\bibfnamefont {A.}~\bibnamefont {Shevchuk}},
  \bibinfo {author} {\bibfnamefont {P.}~\bibnamefont {Novak}}, \bibinfo
  {author} {\bibfnamefont {Y.}~\bibnamefont {Korchev}},\ and\ \bibinfo {author}
  {\bibfnamefont {I.}~\bibnamefont {Parmryd}},\ }\href@noop {} {\bibfield
  {journal} {\bibinfo  {journal} {Nat. Methods}\ }\textbf {\bibinfo {volume}
  {7}},\ \bibinfo {pages} {170} (\bibinfo {year} {2010})}\BibitemShut {NoStop}%
\bibitem [{\citenamefont {Aimon}\ \emph {et~al.}(2013)\citenamefont {Aimon},
  \citenamefont {Callan-Jones}, \citenamefont {Berthaud}, \citenamefont
  {Pinot}, \citenamefont {Toombes},\ and\ \citenamefont
  {Bassereau}}]{BassereauShapeProteinDistribution2013}%
  \BibitemOpen
  \bibfield  {author} {\bibinfo {author} {\bibfnamefont {S.}~\bibnamefont
  {Aimon}}, \bibinfo {author} {\bibfnamefont {A.}~\bibnamefont {Callan-Jones}},
  \bibinfo {author} {\bibfnamefont {A.}~\bibnamefont {Berthaud}}, \bibinfo
  {author} {\bibfnamefont {M.}~\bibnamefont {Pinot}}, \bibinfo {author}
  {\bibfnamefont {G.}~\bibnamefont {Toombes}},\ and\ \bibinfo {author}
  {\bibfnamefont {P.}~\bibnamefont {Bassereau}},\ }\href@noop {} {\bibfield
  {journal} {\bibinfo  {journal} {Preprint}\ }\textbf {\bibinfo {volume}
  {000}},\ \bibinfo {pages} {0000} (\bibinfo {year} {2013})}\BibitemShut
  {NoStop}%
\bibitem [{\citenamefont {Vogel}\ and\ \citenamefont
  {Sheetz}(2006)}]{Vogel2006}%
  \BibitemOpen
  \bibfield  {author} {\bibinfo {author} {\bibfnamefont {V.}~\bibnamefont
  {Vogel}}\ and\ \bibinfo {author} {\bibfnamefont {M.}~\bibnamefont {Sheetz}},\
  }\href@noop {} {\bibfield  {journal} {\bibinfo  {journal} {Nat. Rev. Mol.
  Cell Biol.}\ }\textbf {\bibinfo {volume} {7}},\ \bibinfo {pages} {265}
  (\bibinfo {year} {2006})}\BibitemShut {NoStop}%
\bibitem [{\citenamefont {Kusters}\ and\ \citenamefont
  {Storm}(2014)}]{Kusters2014}%
  \BibitemOpen
  \bibfield  {author} {\bibinfo {author} {\bibfnamefont {R.}~\bibnamefont
  {Kusters}}\ and\ \bibinfo {author} {\bibfnamefont {C.}~\bibnamefont
  {Storm}},\ }\href@noop {} {\bibfield  {journal} {\bibinfo  {journal}
  {Physical Review E}\ }\textbf {\bibinfo {volume} {89}},\ \bibinfo {pages}
  {032723} (\bibinfo {year} {2014})}\BibitemShut {NoStop}%
\bibitem [{\citenamefont {Voeltz}\ and\ \citenamefont
  {Prinz}(2007)}]{Voeltz2007}%
  \BibitemOpen
  \bibfield  {author} {\bibinfo {author} {\bibfnamefont {G.~K.}\ \bibnamefont
  {Voeltz}}\ and\ \bibinfo {author} {\bibfnamefont {W.~A.}\ \bibnamefont
  {Prinz}},\ }\href {http://dx.doi.org/10.1038/nrm2119} {\bibfield  {journal}
  {\bibinfo  {journal} {Nat Rev Mol Cell Biol}\ }\textbf {\bibinfo {volume}
  {8}},\ \bibinfo {pages} {258} (\bibinfo {year} {2007})}\BibitemShut {NoStop}%
\bibitem [{\citenamefont {Parthasarathy}\ and\ \citenamefont
  {Groves}(2007)}]{Groves2007}%
  \BibitemOpen
  \bibfield  {author} {\bibinfo {author} {\bibfnamefont {R.}~\bibnamefont
  {Parthasarathy}}\ and\ \bibinfo {author} {\bibfnamefont {J.~T.}\ \bibnamefont
  {Groves}},\ }\href {http://dx.doi.org/10.1039/B608631D} {\bibfield  {journal}
  {\bibinfo  {journal} {Soft Matter}\ }\textbf {\bibinfo {volume} {3}},\
  \bibinfo {pages} {24} (\bibinfo {year} {2007})}\BibitemShut {NoStop}%
\bibitem [{\citenamefont {Powers}\ \emph {et~al.}(2002)\citenamefont {Powers},
  \citenamefont {Huber},\ and\ \citenamefont {Goldstein}}]{Powers2002}%
  \BibitemOpen
  \bibfield  {author} {\bibinfo {author} {\bibfnamefont {T.~R.}\ \bibnamefont
  {Powers}}, \bibinfo {author} {\bibfnamefont {G.}~\bibnamefont {Huber}},\ and\
  \bibinfo {author} {\bibfnamefont {R.~E.}\ \bibnamefont {Goldstein}},\ }\href
  {http://link.aps.org/doi/10.1103/PhysRevE.65.041901} {\bibfield  {journal}
  {\bibinfo  {journal} {Phys. Rev. E}\ }\textbf {\bibinfo {volume} {65}},\
  \bibinfo {pages} {041901} (\bibinfo {year} {2002})}\BibitemShut {NoStop}%
\bibitem [{\citenamefont {Kahraman}\ \emph {et~al.}(2012)\citenamefont
  {Kahraman}, \citenamefont {Stoop},\ and\ \citenamefont
  {Müller}}]{Muller2012}%
  \BibitemOpen
  \bibfield  {author} {\bibinfo {author} {\bibfnamefont {O.}~\bibnamefont
  {Kahraman}}, \bibinfo {author} {\bibfnamefont {N.}~\bibnamefont {Stoop}},\
  and\ \bibinfo {author} {\bibfnamefont {M.~M.}\ \bibnamefont {Müller}},\
  }\href {http://stacks.iop.org/1367-2630/14/i=9/a=095021} {\bibfield
  {journal} {\bibinfo  {journal} {New Journal of Physics}\ }\textbf {\bibinfo
  {volume} {14}},\ \bibinfo {pages} {095021} (\bibinfo {year}
  {2012})}\BibitemShut {NoStop}%
\bibitem [{\citenamefont {David~Nelson}(2004)}]{NelsonStatMechMem2004}%
  \BibitemOpen
  \bibfield  {author} {\bibinfo {author} {\bibfnamefont {T.~P.}\ \bibnamefont
  {David~Nelson}, \bibfnamefont {Steven~Weinberg}},\ }\href@noop {} {\emph
  {\bibinfo {title} {Statistical Mechanics of Membranes and Surfaces}}}\
  (\bibinfo  {publisher} {World Scientific Publishing},\ \bibinfo {year}
  {2004})\BibitemShut {NoStop}%
\bibitem [{\citenamefont {Saffman}\ and\ \citenamefont
  {Delbr\"{u}ck}(1975)}]{Saffman1975}%
  \BibitemOpen
  \bibfield  {author} {\bibinfo {author} {\bibfnamefont {P.~G.}\ \bibnamefont
  {Saffman}}\ and\ \bibinfo {author} {\bibfnamefont {M.}~\bibnamefont
  {Delbr\"{u}ck}},\ }\href@noop {} {\bibfield  {journal} {\bibinfo  {journal}
  {Proc. Natl. Acad. Sci. USA}\ }\textbf {\bibinfo {volume} {72}},\ \bibinfo
  {pages} {3111} (\bibinfo {year} {1975})}\BibitemShut {NoStop}%
\bibitem [{\citenamefont {Peters}\ and\ \citenamefont
  {Cherry}(1982)}]{Peters1982}%
  \BibitemOpen
  \bibfield  {author} {\bibinfo {author} {\bibfnamefont {R.}~\bibnamefont
  {Peters}}\ and\ \bibinfo {author} {\bibfnamefont {R.~J.}\ \bibnamefont
  {Cherry}},\ }\href {http://www.pnas.org/content/79/14/4317.abstract}
  {\bibfield  {journal} {\bibinfo  {journal} {Proceedings of the National
  Academy of Sciences}\ }\textbf {\bibinfo {volume} {79}},\ \bibinfo {pages}
  {4317} (\bibinfo {year} {1982})}\BibitemShut {NoStop}%
\bibitem [{\citenamefont {Domanov}\ \emph {et~al.}(2011)\citenamefont
  {Domanov}, \citenamefont {Aimon}, \citenamefont {Toombes}, \citenamefont
  {Renner}, \citenamefont {Quemeneur}, \citenamefont {Triller}, \citenamefont
  {Turner},\ and\ \citenamefont
  {Bassereau}}]{BassereauMobilityConfinement2011}%
  \BibitemOpen
  \bibfield  {author} {\bibinfo {author} {\bibfnamefont {Y.~A.}\ \bibnamefont
  {Domanov}}, \bibinfo {author} {\bibfnamefont {S.}~\bibnamefont {Aimon}},
  \bibinfo {author} {\bibfnamefont {G.~E.~S.}\ \bibnamefont {Toombes}},
  \bibinfo {author} {\bibfnamefont {M.}~\bibnamefont {Renner}}, \bibinfo
  {author} {\bibfnamefont {F.}~\bibnamefont {Quemeneur}}, \bibinfo {author}
  {\bibfnamefont {A.}~\bibnamefont {Triller}}, \bibinfo {author} {\bibfnamefont
  {M.~S.}\ \bibnamefont {Turner}},\ and\ \bibinfo {author} {\bibfnamefont
  {P.}~\bibnamefont {Bassereau}},\ }\href@noop {} {\bibfield  {journal}
  {\bibinfo  {journal} {Proc. Natl. Acad. Sci. USA}\ }\textbf {\bibinfo
  {volume} {108}},\ \bibinfo {pages} {12605} (\bibinfo {year}
  {2011})}\BibitemShut {NoStop}%
\bibitem [{\citenamefont {Quemeneur}\ \emph {et~al.}(2014)\citenamefont
  {Quemeneur}, \citenamefont {Sigurdsson}, \citenamefont {Renner},
  \citenamefont {Atzberger}, \citenamefont {Bassereau},\ and\ \citenamefont
  {Lacoste}}]{AtzbergerBassereau2014}%
  \BibitemOpen
  \bibfield  {author} {\bibinfo {author} {\bibfnamefont {F.}~\bibnamefont
  {Quemeneur}}, \bibinfo {author} {\bibfnamefont {J.~K.}\ \bibnamefont
  {Sigurdsson}}, \bibinfo {author} {\bibfnamefont {M.}~\bibnamefont {Renner}},
  \bibinfo {author} {\bibfnamefont {P.~J.}\ \bibnamefont {Atzberger}}, \bibinfo
  {author} {\bibfnamefont {P.}~\bibnamefont {Bassereau}},\ and\ \bibinfo
  {author} {\bibfnamefont {D.}~\bibnamefont {Lacoste}},\ }\href
  {https://doi.org/10.1073/pnas.1321054111} {\bibfield  {journal} {\bibinfo
  {journal} {Proceedings of the National Academy of Sciences}\ }\textbf
  {\bibinfo {volume} {111}},\ \bibinfo {pages} {5083} (\bibinfo {year}
  {2014})}\BibitemShut {NoStop}%
\bibitem [{\citenamefont {Liu}\ \emph {et~al.}(2017)\citenamefont {Liu},
  \citenamefont {Marple}, \citenamefont {Allard}, \citenamefont {Li},
  \citenamefont {Veerapaneni},\ and\ \citenamefont
  {Lowengrub}}]{Allard_Lowengrub_2017}%
  \BibitemOpen
  \bibfield  {author} {\bibinfo {author} {\bibfnamefont {K.}~\bibnamefont
  {Liu}}, \bibinfo {author} {\bibfnamefont {G.~R.}\ \bibnamefont {Marple}},
  \bibinfo {author} {\bibfnamefont {J.}~\bibnamefont {Allard}}, \bibinfo
  {author} {\bibfnamefont {S.}~\bibnamefont {Li}}, \bibinfo {author}
  {\bibfnamefont {S.}~\bibnamefont {Veerapaneni}},\ and\ \bibinfo {author}
  {\bibfnamefont {J.}~\bibnamefont {Lowengrub}},\ }\href
  {https://doi.org/10.1039/c6sm02452a} {\bibfield  {journal} {\bibinfo
  {journal} {Soft matter}\ }\textbf {\bibinfo {volume} {13}},\ \bibinfo {pages}
  {3521} (\bibinfo {year} {2017})}\BibitemShut {NoStop}%
\bibitem [{\citenamefont {Mahapatra}\ \emph {et~al.}(2021)\citenamefont
  {Mahapatra}, \citenamefont {Saintillan},\ and\ \citenamefont
  {Rangamani}}]{Rangamani2021}%
  \BibitemOpen
  \bibfield  {author} {\bibinfo {author} {\bibfnamefont {A.}~\bibnamefont
  {Mahapatra}}, \bibinfo {author} {\bibfnamefont {D.}~\bibnamefont
  {Saintillan}},\ and\ \bibinfo {author} {\bibfnamefont {P.}~\bibnamefont
  {Rangamani}},\ }\href {https://doi.org/10.1039/D1SM00502B} {\bibfield
  {journal} {\bibinfo  {journal} {Soft Matter}\ ,\ } (\bibinfo {year}
  {2021})}\BibitemShut {NoStop}%
\bibitem [{\citenamefont {Yushutin}\ \emph {et~al.}(2019)\citenamefont
  {Yushutin}, \citenamefont {Quaini}, \citenamefont {Majd},\ and\ \citenamefont
  {Olshanskii}}]{Yushutin2019}%
  \BibitemOpen
  \bibfield  {author} {\bibinfo {author} {\bibfnamefont {V.}~\bibnamefont
  {Yushutin}}, \bibinfo {author} {\bibfnamefont {A.}~\bibnamefont {Quaini}},
  \bibinfo {author} {\bibfnamefont {S.}~\bibnamefont {Majd}},\ and\ \bibinfo
  {author} {\bibfnamefont {M.}~\bibnamefont {Olshanskii}},\ }\href@noop {}
  {\bibfield  {journal} {\bibinfo  {journal} {International journal for
  numerical methods in biomedical engineering}\ }\textbf {\bibinfo {volume}
  {35}},\ \bibinfo {pages} {e3181} (\bibinfo {year} {2019})}\BibitemShut
  {NoStop}%
\bibitem [{\citenamefont {Pastor}(1994)}]{Pastor1994}%
  \BibitemOpen
  \bibfield  {author} {\bibinfo {author} {\bibfnamefont {R.~W.}\ \bibnamefont
  {Pastor}},\ }\href@noop {} {\bibfield  {journal} {\bibinfo  {journal} {Curr.
  Opin. Struct. Biol.}\ }\textbf {\bibinfo {volume} {4}},\ \bibinfo {pages}
  {486} (\bibinfo {year} {1994})}\BibitemShut {NoStop}%
\bibitem [{\citenamefont {Sintes}\ and\ \citenamefont
  {Baumg{\"{a}}rtner}(1998)}]{Sintes1998}%
  \BibitemOpen
  \bibfield  {author} {\bibinfo {author} {\bibfnamefont {T.}~\bibnamefont
  {Sintes}}\ and\ \bibinfo {author} {\bibfnamefont {A.}~\bibnamefont
  {Baumg{\"{a}}rtner}},\ }\href@noop {} {\bibfield  {journal} {\bibinfo
  {journal} {Physica A}\ }\textbf {\bibinfo {volume} {249}},\ \bibinfo {pages}
  {571} (\bibinfo {year} {1998})}\BibitemShut {NoStop}%
\bibitem [{\citenamefont {Kerr}\ \emph {et~al.}(2008)\citenamefont {Kerr},
  \citenamefont {Bartol}, \citenamefont {Kaminsky}, \citenamefont {Dittrich},
  \citenamefont {Chang}, \citenamefont {Baden}, \citenamefont {Sejnowski},\
  and\ \citenamefont {Stiles}}]{Kerr2008}%
  \BibitemOpen
  \bibfield  {author} {\bibinfo {author} {\bibfnamefont {R.~A.}\ \bibnamefont
  {Kerr}}, \bibinfo {author} {\bibfnamefont {T.~M.}\ \bibnamefont {Bartol}},
  \bibinfo {author} {\bibfnamefont {B.}~\bibnamefont {Kaminsky}}, \bibinfo
  {author} {\bibfnamefont {M.}~\bibnamefont {Dittrich}}, \bibinfo {author}
  {\bibfnamefont {J.-C.~J.}\ \bibnamefont {Chang}}, \bibinfo {author}
  {\bibfnamefont {S.~B.}\ \bibnamefont {Baden}}, \bibinfo {author}
  {\bibfnamefont {T.~J.}\ \bibnamefont {Sejnowski}},\ and\ \bibinfo {author}
  {\bibfnamefont {J.~R.}\ \bibnamefont {Stiles}},\ }\href@noop {} {\bibfield
  {journal} {\bibinfo  {journal} {SIAM journal on scientific computing : a
  publication of the Society for Industrial and Applied Mathematics}\ }\textbf
  {\bibinfo {volume} {30}},\ \bibinfo {pages} {3126} (\bibinfo {year}
  {2008})}\BibitemShut {NoStop}%
\bibitem [{\citenamefont {Schöneberg}\ \emph {et~al.}(2014)\citenamefont
  {Schöneberg}, \citenamefont {Ullrich},\ and\ \citenamefont
  {Noé}}]{Schoeneberg2014}%
  \BibitemOpen
  \bibfield  {author} {\bibinfo {author} {\bibfnamefont {J.}~\bibnamefont
  {Schöneberg}}, \bibinfo {author} {\bibfnamefont {A.}~\bibnamefont
  {Ullrich}},\ and\ \bibinfo {author} {\bibfnamefont {F.}~\bibnamefont
  {Noé}},\ }\href {https://doi.org/10.1186/s13628-014-0011-5} {\bibfield
  {journal} {\bibinfo  {journal} {BMC Biophysics}\ }\textbf {\bibinfo {volume}
  {7}},\ \bibinfo {pages} {11} (\bibinfo {year} {2014})}\BibitemShut {NoStop}%
\bibitem [{\citenamefont {Kahraman}\ \emph {et~al.}(2016)\citenamefont
  {Kahraman}, \citenamefont {Li},\ and\ \citenamefont
  {Haselwandter}}]{Kahraman2016}%
  \BibitemOpen
  \bibfield  {author} {\bibinfo {author} {\bibfnamefont {O.}~\bibnamefont
  {Kahraman}}, \bibinfo {author} {\bibfnamefont {Y.}~\bibnamefont {Li}},\ and\
  \bibinfo {author} {\bibfnamefont {C.~A.}\ \bibnamefont {Haselwandter}},\
  }\href {https://doi.org/10.1209/0295-5075/115/68006} {\bibfield  {journal}
  {\bibinfo  {journal} {{EPL} (Europhysics Letters)}\ }\textbf {\bibinfo
  {volume} {115}},\ \bibinfo {pages} {68006} (\bibinfo {year}
  {2016})}\BibitemShut {NoStop}%
\bibitem [{\citenamefont {Mauro}\ \emph {et~al.}(2014)\citenamefont {Mauro},
  \citenamefont {Sigurdsson}, \citenamefont {Shrake}, \citenamefont
  {Atzberger},\ and\ \citenamefont {Isaacson}}]{Atzberger_Isaacson_2014}%
  \BibitemOpen
  \bibfield  {author} {\bibinfo {author} {\bibfnamefont {A.~J.}\ \bibnamefont
  {Mauro}}, \bibinfo {author} {\bibfnamefont {J.~K.}\ \bibnamefont
  {Sigurdsson}}, \bibinfo {author} {\bibfnamefont {J.}~\bibnamefont {Shrake}},
  \bibinfo {author} {\bibfnamefont {P.~J.}\ \bibnamefont {Atzberger}},\ and\
  \bibinfo {author} {\bibfnamefont {S.~A.}\ \bibnamefont {Isaacson}},\ }\href
  {https://doi.org/https://doi.org/10.1016/j.jcp.2013.12.023} {\bibfield
  {journal} {\bibinfo  {journal} {Journal of Computational Physics}\ }\textbf
  {\bibinfo {volume} {259}},\ \bibinfo {pages} {536} (\bibinfo {year}
  {2014})}\BibitemShut {NoStop}%
\bibitem [{\citenamefont {Collins}\ \emph {et~al.}(2010)\citenamefont
  {Collins}, \citenamefont {Stamatakis},\ and\ \citenamefont
  {Vlachos}}]{Collins2010}%
  \BibitemOpen
  \bibfield  {author} {\bibinfo {author} {\bibfnamefont {S.}~\bibnamefont
  {Collins}}, \bibinfo {author} {\bibfnamefont {M.}~\bibnamefont
  {Stamatakis}},\ and\ \bibinfo {author} {\bibfnamefont {D.~G.}\ \bibnamefont
  {Vlachos}},\ }\href {https://doi.org/10.1186/1471-2105-11-218} {\bibfield
  {journal} {\bibinfo  {journal} {BMC Bioinformatics}\ }\textbf {\bibinfo
  {volume} {11}},\ \bibinfo {pages} {218} (\bibinfo {year} {2010})}\BibitemShut
  {NoStop}%
\bibitem [{\citenamefont {Sawhney}\ and\ \citenamefont
  {Crane}(2020)}]{Crane2020}%
  \BibitemOpen
  \bibfield  {author} {\bibinfo {author} {\bibfnamefont {R.}~\bibnamefont
  {Sawhney}}\ and\ \bibinfo {author} {\bibfnamefont {K.}~\bibnamefont
  {Crane}},\ }\href@noop {} {\bibfield  {journal} {\bibinfo  {journal} {ACM
  Transactions on Graphics}\ }\textbf {\bibinfo {volume} {39}} (\bibinfo {year}
  {2020})}\BibitemShut {NoStop}%
\bibitem [{\citenamefont {D.~P.~Tieleman}\ and\ \citenamefont
  {Berendsen}(1997)}]{Tieleman1997}%
  \BibitemOpen
  \bibfield  {author} {\bibinfo {author} {\bibfnamefont {S.~J.~M.}\
  \bibnamefont {D.~P.~Tieleman}}\ and\ \bibinfo {author} {\bibfnamefont
  {H.~J.~C.}\ \bibnamefont {Berendsen}},\ }\href@noop {} {\bibfield  {journal}
  {\bibinfo  {journal} {Biochim. Biophys. Acta.}\ }\textbf {\bibinfo {volume}
  {1331}},\ \bibinfo {pages} {235} (\bibinfo {year} {1997})}\BibitemShut
  {NoStop}%
\bibitem [{\citenamefont {Goossens}\ and\ \citenamefont
  {De~Winter}(2018)}]{Goossens2018}%
  \BibitemOpen
  \bibfield  {author} {\bibinfo {author} {\bibfnamefont {K.}~\bibnamefont
  {Goossens}}\ and\ \bibinfo {author} {\bibfnamefont {H.}~\bibnamefont
  {De~Winter}},\ }\href {https://doi.org/10.1021/acs.jcim.8b00639} {\bibfield
  {journal} {\bibinfo  {journal} {J. Chem. Inf. Model.}\ }\textbf {\bibinfo
  {volume} {58}},\ \bibinfo {pages} {2193} (\bibinfo {year}
  {2018})}\BibitemShut {NoStop}%
\bibitem [{\citenamefont {Grouleff}\ \emph {et~al.}(2015)\citenamefont
  {Grouleff}, \citenamefont {Irudayam}, \citenamefont {Skeby},\ and\
  \citenamefont {Schiøtt}}]{Grouleff2015}%
  \BibitemOpen
  \bibfield  {author} {\bibinfo {author} {\bibfnamefont {J.}~\bibnamefont
  {Grouleff}}, \bibinfo {author} {\bibfnamefont {S.~J.}\ \bibnamefont
  {Irudayam}}, \bibinfo {author} {\bibfnamefont {K.~K.}\ \bibnamefont
  {Skeby}},\ and\ \bibinfo {author} {\bibfnamefont {B.}~\bibnamefont
  {Schiøtt}},\ }\href
  {https://doi.org/https://doi.org/10.1016/j.bbamem.2015.03.029} {\bibfield
  {journal} {\bibinfo  {journal} {Biochimica et Biophysica Acta (BBA) -
  Biomembranes}\ }\textbf {\bibinfo {volume} {1848}},\ \bibinfo {pages} {1783}
  (\bibinfo {year} {2015})},\ \bibinfo {note} {lipid-protein
  interactions}\BibitemShut {NoStop}%
\bibitem [{\citenamefont {Reynwar}\ \emph {et~al.}(2007)\citenamefont
  {Reynwar}, \citenamefont {Illya}, \citenamefont {Harmandaris}, \citenamefont
  {Muller}, \citenamefont {Kremer},\ and\ \citenamefont
  {Deserno}}]{DesernoVirusAggregation2007}%
  \BibitemOpen
  \bibfield  {author} {\bibinfo {author} {\bibfnamefont {B.~J.}\ \bibnamefont
  {Reynwar}}, \bibinfo {author} {\bibfnamefont {G.}~\bibnamefont {Illya}},
  \bibinfo {author} {\bibfnamefont {V.~A.}\ \bibnamefont {Harmandaris}},
  \bibinfo {author} {\bibfnamefont {M.~M.}\ \bibnamefont {Muller}}, \bibinfo
  {author} {\bibfnamefont {K.}~\bibnamefont {Kremer}},\ and\ \bibinfo {author}
  {\bibfnamefont {M.}~\bibnamefont {Deserno}},\ }\href
  {http://dx.doi.org/10.1038/nature05840} {\bibfield  {journal} {\bibinfo
  {journal} {Nature}\ }\textbf {\bibinfo {volume} {447}},\ \bibinfo {pages}
  {461} (\bibinfo {year} {2007})}\BibitemShut {NoStop}%
\bibitem [{\citenamefont {Marrink}\ \emph {et~al.}(2007)\citenamefont
  {Marrink}, \citenamefont {Risselada}, \citenamefont {Yefimov}, \citenamefont
  {Tieleman},\ and\ \citenamefont {de~Vries}}]{Marrink2007}%
  \BibitemOpen
  \bibfield  {author} {\bibinfo {author} {\bibfnamefont {S.~J.}\ \bibnamefont
  {Marrink}}, \bibinfo {author} {\bibfnamefont {H.~J.}\ \bibnamefont
  {Risselada}}, \bibinfo {author} {\bibfnamefont {S.}~\bibnamefont {Yefimov}},
  \bibinfo {author} {\bibfnamefont {D.~P.}\ \bibnamefont {Tieleman}},\ and\
  \bibinfo {author} {\bibfnamefont {A.~H.}\ \bibnamefont {de~Vries}},\ }\href
  {https://doi.org/10.1021/jp071097f} {\bibfield  {journal} {\bibinfo
  {journal} {The Journal of Physical Chemistry B}\ }\textbf {\bibinfo {volume}
  {111}},\ \bibinfo {pages} {7812} (\bibinfo {year} {2007})},\ \bibinfo {note}
  {pMID: 17569554},\ \Eprint
  {https://arxiv.org/abs/http://pubs.acs.org/doi/pdf/10.1021/jp071097f}
  {http://pubs.acs.org/doi/pdf/10.1021/jp071097f} \BibitemShut {NoStop}%
\bibitem [{\citenamefont {Camley}\ and\ \citenamefont
  {Brown}(2010)}]{Camley2010}%
  \BibitemOpen
  \bibfield  {author} {\bibinfo {author} {\bibfnamefont {B.~A.}\ \bibnamefont
  {Camley}}\ and\ \bibinfo {author} {\bibfnamefont {F.~L.~H.}\ \bibnamefont
  {Brown}},\ }\href {http://link.aps.org/doi/10.1103/PhysRevLett.105.148102}
  {\bibfield  {journal} {\bibinfo  {journal} {Phys. Rev. Lett.}\ }\textbf
  {\bibinfo {volume} {105}},\ \bibinfo {pages} {148102} (\bibinfo {year}
  {2010})}\BibitemShut {NoStop}%
\bibitem [{\citenamefont {Naji}\ \emph {et~al.}(2009)\citenamefont {Naji},
  \citenamefont {Atzberger},\ and\ \citenamefont {Brown}}]{AtzbergerNaji2009}%
  \BibitemOpen
  \bibfield  {author} {\bibinfo {author} {\bibfnamefont {A.}~\bibnamefont
  {Naji}}, \bibinfo {author} {\bibfnamefont {P.~J.}\ \bibnamefont
  {Atzberger}},\ and\ \bibinfo {author} {\bibfnamefont {F.~L.~H.}\ \bibnamefont
  {Brown}},\ }\href {http://link.aps.org/doi/10.1103/PhysRevLett.102.138102}
  {\bibfield  {journal} {\bibinfo  {journal} {Phys. Rev. Lett.}\ }\textbf
  {\bibinfo {volume} {102}},\ \bibinfo {pages} {138102} (\bibinfo {year}
  {2009})}\BibitemShut {NoStop}%
\bibitem [{\citenamefont {Wang}\ \emph {et~al.}(2013)\citenamefont {Wang},
  \citenamefont {Sigurdsson}, \citenamefont {Brandt},\ and\ \citenamefont
  {Atzberger}}]{AtzbergerBilayerVesicle2013}%
  \BibitemOpen
  \bibfield  {author} {\bibinfo {author} {\bibfnamefont {Y.}~\bibnamefont
  {Wang}}, \bibinfo {author} {\bibfnamefont {J.~K.}\ \bibnamefont
  {Sigurdsson}}, \bibinfo {author} {\bibfnamefont {E.}~\bibnamefont {Brandt}},\
  and\ \bibinfo {author} {\bibfnamefont {P.~J.}\ \bibnamefont {Atzberger}},\
  }\href {http://link.aps.org/doi/10.1103/PhysRevE.88.023301} {\bibfield
  {journal} {\bibinfo  {journal} {Phys. Rev. E}\ }\textbf {\bibinfo {volume}
  {88}},\ \bibinfo {pages} {023301} (\bibinfo {year} {2013})}\BibitemShut
  {NoStop}%
\bibitem [{\citenamefont {Sigurdsson}\ \emph {et~al.}(2013)\citenamefont
  {Sigurdsson}, \citenamefont {Brown},\ and\ \citenamefont
  {Atzberger}}]{AtzbergerSigurdsson2012}%
  \BibitemOpen
  \bibfield  {author} {\bibinfo {author} {\bibfnamefont {J.~K.}\ \bibnamefont
  {Sigurdsson}}, \bibinfo {author} {\bibfnamefont {F.~L.}\ \bibnamefont
  {Brown}},\ and\ \bibinfo {author} {\bibfnamefont {P.~J.}\ \bibnamefont
  {Atzberger}},\ }\href
  {http://www.sciencedirect.com/science/article/pii/S0021999113004403}
  {\bibfield  {journal} {\bibinfo  {journal} {Journal of Computational
  Physics}\ }\textbf {\bibinfo {volume} {252}},\ \bibinfo {pages} {65}
  (\bibinfo {year} {2013})}\BibitemShut {NoStop}%
\bibitem [{\citenamefont {Reister}\ and\ \citenamefont
  {Seifert}(2005)}]{Reister2005}%
  \BibitemOpen
  \bibfield  {author} {\bibinfo {author} {\bibfnamefont {E.}~\bibnamefont
  {Reister}}\ and\ \bibinfo {author} {\bibfnamefont {U.}~\bibnamefont
  {Seifert}},\ }\href@noop {} {\bibfield  {journal} {\bibinfo  {journal} {Eur.
  Phys. Lett.}\ }\textbf {\bibinfo {volume} {71}},\ \bibinfo {pages} {859}
  (\bibinfo {year} {2005})}\BibitemShut {NoStop}%
\bibitem [{\citenamefont {Reister-Gottfried}\ \emph {et~al.}(2010)\citenamefont
  {Reister-Gottfried}, \citenamefont {Leitenberger},\ and\ \citenamefont
  {Seifert}}]{ReisterGottfried2010}%
  \BibitemOpen
  \bibfield  {author} {\bibinfo {author} {\bibfnamefont {E.}~\bibnamefont
  {Reister-Gottfried}}, \bibinfo {author} {\bibfnamefont {S.~M.}\ \bibnamefont
  {Leitenberger}},\ and\ \bibinfo {author} {\bibfnamefont {U.}~\bibnamefont
  {Seifert}},\ }\href@noop {} {\bibfield  {journal} {\bibinfo  {journal} {Phys.
  Rev. E}\ }\textbf {\bibinfo {volume} {81}},\ \bibinfo {pages} {031903}
  (\bibinfo {year} {2010})}\BibitemShut {NoStop}%
\bibitem [{\citenamefont {Oppenheimer}\ and\ \citenamefont
  {Diamant}(2009)}]{Oppenheimer2009}%
  \BibitemOpen
  \bibfield  {author} {\bibinfo {author} {\bibfnamefont {N.}~\bibnamefont
  {Oppenheimer}}\ and\ \bibinfo {author} {\bibfnamefont {H.}~\bibnamefont
  {Diamant}},\ }\href
  {http://www.sciencedirect.com/science/article/pii/S0006349509004214}
  {\bibfield  {journal} {\bibinfo  {journal} {Biophysical Journal}\ }\textbf
  {\bibinfo {volume} {96}},\ \bibinfo {pages} {3041} (\bibinfo {year}
  {2009})}\BibitemShut {NoStop}%
\bibitem [{\citenamefont {Rower}\ \emph {et~al.}(2019)\citenamefont {Rower},
  \citenamefont {Padidar},\ and\ \citenamefont
  {Atzberger}}]{Atzberger_Surf_Fluct_2019}%
  \BibitemOpen
  \bibfield  {author} {\bibinfo {author} {\bibfnamefont {D.}~\bibnamefont
  {Rower}}, \bibinfo {author} {\bibfnamefont {M.}~\bibnamefont {Padidar}},\
  and\ \bibinfo {author} {\bibfnamefont {P.~J.}\ \bibnamefont {Atzberger}},\
  }\href {https://arxiv.org/abs/1906.01146} {\bibfield  {journal} {\bibinfo
  {journal} {arXiv}\ } (\bibinfo {year} {2019})}\BibitemShut {NoStop}%
\bibitem [{\citenamefont {Macdonald}\ \emph {et~al.}(2013)\citenamefont
  {Macdonald}, \citenamefont {Merriman},\ and\ \citenamefont
  {Ruuth}}]{Macdonald2013}%
  \BibitemOpen
  \bibfield  {author} {\bibinfo {author} {\bibfnamefont {C.~B.}\ \bibnamefont
  {Macdonald}}, \bibinfo {author} {\bibfnamefont {B.}~\bibnamefont
  {Merriman}},\ and\ \bibinfo {author} {\bibfnamefont {S.~J.}\ \bibnamefont
  {Ruuth}},\ }\href {https://doi.org/10.1073/pnas.1221408110} {\bibfield
  {journal} {\bibinfo  {journal} {Proceedings of the National Academy of
  Sciences}\ }\textbf {\bibinfo {volume} {110}},\ \bibinfo {pages} {9209}
  (\bibinfo {year} {2013})},\ \Eprint
  {https://arxiv.org/abs/https://www.pnas.org/content/110/23/9209.full.pdf}
  {https://www.pnas.org/content/110/23/9209.full.pdf} \BibitemShut {NoStop}%
\bibitem [{\citenamefont {Gross}\ and\ \citenamefont
  {Atzberger}(2018)}]{AtzbergerGrossHydroSurf2018}%
  \BibitemOpen
  \bibfield  {author} {\bibinfo {author} {\bibfnamefont {B.}~\bibnamefont
  {Gross}}\ and\ \bibinfo {author} {\bibfnamefont {P.}~\bibnamefont
  {Atzberger}},\ }\href
  {https://doi.org/https://doi.org/10.1016/j.jcp.2018.06.013} {\bibfield
  {journal} {\bibinfo  {journal} {Journal of Computational Physics}\ }\textbf
  {\bibinfo {volume} {371}},\ \bibinfo {pages} {663 } (\bibinfo {year}
  {2018})}\BibitemShut {NoStop}%
\bibitem [{\citenamefont {Lai}\ and\ \citenamefont {Li}(2017)}]{Lai2017}%
  \BibitemOpen
  \bibfield  {author} {\bibinfo {author} {\bibfnamefont {R.}~\bibnamefont
  {Lai}}\ and\ \bibinfo {author} {\bibfnamefont {J.}~\bibnamefont {Li}},\
  }\href@noop {} {\bibfield  {journal} {\bibinfo  {journal} {SIAM Journal on
  Scientific Computing}\ }\textbf {\bibinfo {volume} {39}},\ \bibinfo {pages}
  {A2231} (\bibinfo {year} {2017})}\BibitemShut {NoStop}%
\bibitem [{\citenamefont {Gross}\ \emph {et~al.}(2021)\citenamefont {Gross},
  \citenamefont {Kuberry},\ and\ \citenamefont
  {Atzberger}}]{Atzberger_Surface_FPT_2021}%
  \BibitemOpen
  \bibfield  {author} {\bibinfo {author} {\bibfnamefont {B.~J.}\ \bibnamefont
  {Gross}}, \bibinfo {author} {\bibfnamefont {P.}~\bibnamefont {Kuberry}},\
  and\ \bibinfo {author} {\bibfnamefont {P.~J.}\ \bibnamefont {Atzberger}},\
  }\href@noop {} {\bibfield  {journal} {\bibinfo  {journal} {arXiv preprint
  arXiv:2102.02421}\ } (\bibinfo {year} {2021})}\BibitemShut {NoStop}%
\bibitem [{\citenamefont {Alvarez}\ and\ \citenamefont
  {Sabatini}(2007)}]{Alvarez2007}%
  \BibitemOpen
  \bibfield  {author} {\bibinfo {author} {\bibfnamefont {V.~A.}\ \bibnamefont
  {Alvarez}}\ and\ \bibinfo {author} {\bibfnamefont {B.~L.}\ \bibnamefont
  {Sabatini}},\ }\href {https://doi.org/10.1146/annurev.neuro.30.051606.094222}
  {\bibfield  {journal} {\bibinfo  {journal} {Annual Review of Neuroscience}\
  }\textbf {\bibinfo {volume} {30}},\ \bibinfo {pages} {79} (\bibinfo {year}
  {2007})},\ \bibinfo {note} {pMID: 17280523},\ \Eprint
  {https://arxiv.org/abs/https://doi.org/10.1146/annurev.neuro.30.051606.094222}
  {https://doi.org/10.1146/annurev.neuro.30.051606.094222} \BibitemShut
  {NoStop}%
\bibitem [{\citenamefont {Yuste}\ and\ \citenamefont {Denk}(1995)}]{Yuste1995}%
  \BibitemOpen
  \bibfield  {author} {\bibinfo {author} {\bibfnamefont {R.}~\bibnamefont
  {Yuste}}\ and\ \bibinfo {author} {\bibfnamefont {W.}~\bibnamefont {Denk}},\
  }\href {https://doi.org/10.1038/375682a0} {\bibfield  {journal} {\bibinfo
  {journal} {Nature}\ }\textbf {\bibinfo {volume} {375}},\ \bibinfo {pages}
  {682} (\bibinfo {year} {1995})}\BibitemShut {NoStop}%
\bibitem [{\citenamefont {Lu}\ \emph {et~al.}(2014)\citenamefont {Lu},
  \citenamefont {MacGillavry}, \citenamefont {Frost},\ and\ \citenamefont
  {Blanpied}}]{BlanpiedCaMKII2014}%
  \BibitemOpen
  \bibfield  {author} {\bibinfo {author} {\bibfnamefont {H.~E.}\ \bibnamefont
  {Lu}}, \bibinfo {author} {\bibfnamefont {H.~D.}\ \bibnamefont {MacGillavry}},
  \bibinfo {author} {\bibfnamefont {N.~A.}\ \bibnamefont {Frost}},\ and\
  \bibinfo {author} {\bibfnamefont {T.~A.}\ \bibnamefont {Blanpied}},\ }\href
  {https://doi.org/10.1523/JNEUROSCI.4364-13.2014} {\bibfield  {journal}
  {\bibinfo  {journal} {The Journal of Neuroscience}\ }\textbf {\bibinfo
  {volume} {34}},\ \bibinfo {pages} {7600} (\bibinfo {year}
  {2014})}\BibitemShut {NoStop}%
\bibitem [{\citenamefont {Herring}\ and\ \citenamefont
  {Nicoll}(2016)}]{Herring2016}%
  \BibitemOpen
  \bibfield  {author} {\bibinfo {author} {\bibfnamefont {B.~E.}\ \bibnamefont
  {Herring}}\ and\ \bibinfo {author} {\bibfnamefont {R.~A.}\ \bibnamefont
  {Nicoll}},\ }\href {https://doi.org/10.1146/annurev-physiol-021014-071753}
  {\bibfield  {journal} {\bibinfo  {journal} {Annual Review of Physiology}\
  }\textbf {\bibinfo {volume} {78}},\ \bibinfo {pages} {351} (\bibinfo {year}
  {2016})},\ \bibinfo {note} {pMID: 26863325},\ \Eprint
  {https://arxiv.org/abs/https://doi.org/10.1146/annurev-physiol-021014-071753}
  {https://doi.org/10.1146/annurev-physiol-021014-071753} \BibitemShut
  {NoStop}%
\bibitem [{\citenamefont {Nicoll}(2017)}]{Nicoll2017}%
  \BibitemOpen
  \bibfield  {author} {\bibinfo {author} {\bibfnamefont {R.~A.}\ \bibnamefont
  {Nicoll}},\ }\href
  {https://doi.org/https://doi.org/10.1016/j.neuron.2016.12.015} {\bibfield
  {journal} {\bibinfo  {journal} {Neuron}\ }\textbf {\bibinfo {volume} {93}},\
  \bibinfo {pages} {281} (\bibinfo {year} {2017})}\BibitemShut {NoStop}%
\bibitem [{\citenamefont {Holcman}\ and\ \citenamefont
  {Schuss}(2011)}]{Holcman2011}%
  \BibitemOpen
  \bibfield  {author} {\bibinfo {author} {\bibfnamefont {D.}~\bibnamefont
  {Holcman}}\ and\ \bibinfo {author} {\bibfnamefont {Z.}~\bibnamefont
  {Schuss}},\ }\href {http://www.mathematical-neuroscience.com/content/1/1/10}
  {\bibfield  {journal} {\bibinfo  {journal} {The Journal of Mathematical
  Neuroscience}\ }\textbf {\bibinfo {volume} {1}},\ \bibinfo {pages} {10}
  (\bibinfo {year} {2011})}\BibitemShut {NoStop}%
\bibitem [{\citenamefont {Li}\ \emph {et~al.}(2016)\citenamefont {Li},
  \citenamefont {Song}, \citenamefont {MacGillavry}, \citenamefont {Blanpied},\
  and\ \citenamefont {Raghavachari}}]{Li2016}%
  \BibitemOpen
  \bibfield  {author} {\bibinfo {author} {\bibfnamefont {T.~P.}\ \bibnamefont
  {Li}}, \bibinfo {author} {\bibfnamefont {Y.}~\bibnamefont {Song}}, \bibinfo
  {author} {\bibfnamefont {H.~D.}\ \bibnamefont {MacGillavry}}, \bibinfo
  {author} {\bibfnamefont {T.~A.}\ \bibnamefont {Blanpied}},\ and\ \bibinfo
  {author} {\bibfnamefont {S.}~\bibnamefont {Raghavachari}},\ }\href
  {https://doi.org/10.1523/JNEUROSCI.3154-15.2016} {\bibfield  {journal}
  {\bibinfo  {journal} {Journal of Neuroscience}\ }\textbf {\bibinfo {volume}
  {36}},\ \bibinfo {pages} {4276} (\bibinfo {year} {2016})},\ \Eprint
  {https://arxiv.org/abs/https://www.jneurosci.org/content/36/15/4276.full.pdf}
  {https://www.jneurosci.org/content/36/15/4276.full.pdf} \BibitemShut
  {NoStop}%
\bibitem [{\citenamefont {Wang}\ \emph {et~al.}(2016)\citenamefont {Wang},
  \citenamefont {Dumoulin}, \citenamefont {Renner}, \citenamefont {Triller},\
  and\ \citenamefont {Specht}}]{Wang2016}%
  \BibitemOpen
  \bibfield  {author} {\bibinfo {author} {\bibfnamefont {L.}~\bibnamefont
  {Wang}}, \bibinfo {author} {\bibfnamefont {A.}~\bibnamefont {Dumoulin}},
  \bibinfo {author} {\bibfnamefont {M.}~\bibnamefont {Renner}}, \bibinfo
  {author} {\bibfnamefont {A.}~\bibnamefont {Triller}},\ and\ \bibinfo {author}
  {\bibfnamefont {C.~G.}\ \bibnamefont {Specht}},\ }\href
  {https://doi.org/10.1371/journal.pone.0148310} {\bibfield  {journal}
  {\bibinfo  {journal} {PLOS ONE}\ }\textbf {\bibinfo {volume} {11}},\ \bibinfo
  {pages} {1} (\bibinfo {year} {2016})}\BibitemShut {NoStop}%
\bibitem [{\citenamefont {Adrian}\ \emph {et~al.}(2017)\citenamefont {Adrian},
  \citenamefont {Kusters}, \citenamefont {Storm}, \citenamefont {Hoogenraad},\
  and\ \citenamefont {Kapitein}}]{Adrian2017}%
  \BibitemOpen
  \bibfield  {author} {\bibinfo {author} {\bibfnamefont {M.}~\bibnamefont
  {Adrian}}, \bibinfo {author} {\bibfnamefont {R.}~\bibnamefont {Kusters}},
  \bibinfo {author} {\bibfnamefont {C.}~\bibnamefont {Storm}}, \bibinfo
  {author} {\bibfnamefont {C.~C.}\ \bibnamefont {Hoogenraad}},\ and\ \bibinfo
  {author} {\bibfnamefont {L.~C.}\ \bibnamefont {Kapitein}},\ }\href
  {https://doi.org/https://doi.org/10.1016/j.bpj.2017.06.048} {\bibfield
  {journal} {\bibinfo  {journal} {Biophysical Journal}\ }\textbf {\bibinfo
  {volume} {113}},\ \bibinfo {pages} {2261} (\bibinfo {year}
  {2017})}\BibitemShut {NoStop}%
\bibitem [{\citenamefont {Chen}\ and\ \citenamefont
  {De~Schutter}(2017)}]{Chen2017}%
  \BibitemOpen
  \bibfield  {author} {\bibinfo {author} {\bibfnamefont {W.}~\bibnamefont
  {Chen}}\ and\ \bibinfo {author} {\bibfnamefont {E.}~\bibnamefont
  {De~Schutter}},\ }\href {https://doi.org/10.1007/s12021-016-9321-x}
  {\bibfield  {journal} {\bibinfo  {journal} {Neuroinformatics}\ }\textbf
  {\bibinfo {volume} {15}},\ \bibinfo {pages} {1} (\bibinfo {year}
  {2017})}\BibitemShut {NoStop}%
\bibitem [{\citenamefont {Simon}\ \emph {et~al.}(2014)\citenamefont {Simon},
  \citenamefont {Hepburn}, \citenamefont {Chen},\ and\ \citenamefont
  {De~Schutter}}]{Simon2014}%
  \BibitemOpen
  \bibfield  {author} {\bibinfo {author} {\bibfnamefont {C.}~\bibnamefont
  {Simon}}, \bibinfo {author} {\bibfnamefont {I.}~\bibnamefont {Hepburn}},
  \bibinfo {author} {\bibfnamefont {W.}~\bibnamefont {Chen}},\ and\ \bibinfo
  {author} {\bibfnamefont {E.}~\bibnamefont {De~Schutter}},\ }\href
  {http://dx.doi.org/10.1007/s10827-013-0482-4} {\bibfield  {journal} {\bibinfo
   {journal} {Journal of Computational Neuroscience}\ }\textbf {\bibinfo
  {volume} {36}},\ \bibinfo {pages} {483} (\bibinfo {year} {2014})}\BibitemShut
  {NoStop}%
\bibitem [{\citenamefont {Cartailler}\ \emph {et~al.}(2018)\citenamefont
  {Cartailler}, \citenamefont {Kwon}, \citenamefont {Yuste},\ and\
  \citenamefont {Holcman}}]{Cartailler2018}%
  \BibitemOpen
  \bibfield  {author} {\bibinfo {author} {\bibfnamefont {J.}~\bibnamefont
  {Cartailler}}, \bibinfo {author} {\bibfnamefont {T.}~\bibnamefont {Kwon}},
  \bibinfo {author} {\bibfnamefont {R.}~\bibnamefont {Yuste}},\ and\ \bibinfo
  {author} {\bibfnamefont {D.}~\bibnamefont {Holcman}},\ }\href
  {https://doi.org/https://doi.org/10.1016/j.neuron.2018.01.034} {\bibfield
  {journal} {\bibinfo  {journal} {Neuron}\ }\textbf {\bibinfo {volume} {97}},\
  \bibinfo {pages} {1126} (\bibinfo {year} {2018})}\BibitemShut {NoStop}%
\bibitem [{\citenamefont {Cugno}\ \emph {et~al.}(2019)\citenamefont {Cugno},
  \citenamefont {Bartol}, \citenamefont {Sejnowski}, \citenamefont {Iyengar},\
  and\ \citenamefont {Rangamani}}]{Rangamani2019}%
  \BibitemOpen
  \bibfield  {author} {\bibinfo {author} {\bibfnamefont {A.}~\bibnamefont
  {Cugno}}, \bibinfo {author} {\bibfnamefont {T.~M.}\ \bibnamefont {Bartol}},
  \bibinfo {author} {\bibfnamefont {T.~J.}\ \bibnamefont {Sejnowski}}, \bibinfo
  {author} {\bibfnamefont {R.}~\bibnamefont {Iyengar}},\ and\ \bibinfo {author}
  {\bibfnamefont {P.}~\bibnamefont {Rangamani}},\ }\href
  {https://doi.org/10.1038/s41598-019-48028-0} {\bibfield  {journal} {\bibinfo
  {journal} {Scientific Reports}\ }\textbf {\bibinfo {volume} {9}},\ \bibinfo
  {pages} {11676} (\bibinfo {year} {2019})}\BibitemShut {NoStop}%
\bibitem [{\citenamefont {Borczyk}\ \emph {et~al.}(2019)\citenamefont
  {Borczyk}, \citenamefont {Śliwińska}, \citenamefont {Caly}, \citenamefont
  {Bernas},\ and\ \citenamefont {Radwanska}}]{Borczyk2019}%
  \BibitemOpen
  \bibfield  {author} {\bibinfo {author} {\bibfnamefont {M.}~\bibnamefont
  {Borczyk}}, \bibinfo {author} {\bibfnamefont {M.~A.}\ \bibnamefont
  {Śliwińska}}, \bibinfo {author} {\bibfnamefont {A.}~\bibnamefont {Caly}},
  \bibinfo {author} {\bibfnamefont {T.}~\bibnamefont {Bernas}},\ and\ \bibinfo
  {author} {\bibfnamefont {K.}~\bibnamefont {Radwanska}},\ }\href
  {https://doi.org/10.1038/s41598-018-38412-7} {\bibfield  {journal} {\bibinfo
  {journal} {Scientific Reports}\ }\textbf {\bibinfo {volume} {9}},\ \bibinfo
  {pages} {1693} (\bibinfo {year} {2019})}\BibitemShut {NoStop}%
\bibitem [{\citenamefont {Tapia}\ \emph {et~al.}(2019)\citenamefont {Tapia},
  \citenamefont {Saglam}, \citenamefont {Czech}, \citenamefont {Kuczewski},
  \citenamefont {Bartol}, \citenamefont {Sejnowski},\ and\ \citenamefont
  {Faeder}}]{Tapia2019}%
  \BibitemOpen
  \bibfield  {author} {\bibinfo {author} {\bibfnamefont {J.-J.}\ \bibnamefont
  {Tapia}}, \bibinfo {author} {\bibfnamefont {A.~S.}\ \bibnamefont {Saglam}},
  \bibinfo {author} {\bibfnamefont {J.}~\bibnamefont {Czech}}, \bibinfo
  {author} {\bibfnamefont {R.}~\bibnamefont {Kuczewski}}, \bibinfo {author}
  {\bibfnamefont {T.~M.}\ \bibnamefont {Bartol}}, \bibinfo {author}
  {\bibfnamefont {T.~J.}\ \bibnamefont {Sejnowski}},\ and\ \bibinfo {author}
  {\bibfnamefont {J.~R.}\ \bibnamefont {Faeder}},\ }\bibinfo {title} {Mcell-r:
  A particle-resolution network-free spatial modeling framework},\ in\ \href
  {https://doi.org/10.1007/978-1-4939-9102-0_9} {\emph {\bibinfo {booktitle}
  {Modeling Biomolecular Site Dynamics: Methods and Protocols}}},\ \bibinfo
  {editor} {edited by\ \bibinfo {editor} {\bibfnamefont {W.~S.}\ \bibnamefont
  {Hlavacek}}}\ (\bibinfo  {publisher} {Springer New York},\ \bibinfo {address}
  {New York, NY},\ \bibinfo {year} {2019})\ pp.\ \bibinfo {pages}
  {203--229}\BibitemShut {NoStop}%
\bibitem [{\citenamefont {Miermans}\ \emph {et~al.}(2017)\citenamefont
  {Miermans}, \citenamefont {Kusters}, \citenamefont {Hoogenraad},\ and\
  \citenamefont {Storm}}]{Miermans2017}%
  \BibitemOpen
  \bibfield  {author} {\bibinfo {author} {\bibfnamefont {C.~A.}\ \bibnamefont
  {Miermans}}, \bibinfo {author} {\bibfnamefont {R.~P.~T.}\ \bibnamefont
  {Kusters}}, \bibinfo {author} {\bibfnamefont {C.~C.}\ \bibnamefont
  {Hoogenraad}},\ and\ \bibinfo {author} {\bibfnamefont {C.}~\bibnamefont
  {Storm}},\ }\href {https://doi.org/10.1371/journal.pone.0170113} {\bibfield
  {journal} {\bibinfo  {journal} {PLOS ONE}\ }\textbf {\bibinfo {volume}
  {12}},\ \bibinfo {pages} {1} (\bibinfo {year} {2017})}\BibitemShut {NoStop}%
\bibitem [{\citenamefont {Tønnesen}\ and\ \citenamefont
  {Nägerl}(2016)}]{Toennesen2016}%
  \BibitemOpen
  \bibfield  {author} {\bibinfo {author} {\bibfnamefont {J.}~\bibnamefont
  {Tønnesen}}\ and\ \bibinfo {author} {\bibfnamefont {U.~V.}\ \bibnamefont
  {Nägerl}},\ }\href {https://doi.org/10.3389/fpsyt.2016.00101} {\bibfield
  {journal} {\bibinfo  {journal} {Frontiers in Psychiatry}\ }\textbf {\bibinfo
  {volume} {7}},\ \bibinfo {pages} {101} (\bibinfo {year} {2016})}\BibitemShut
  {NoStop}%
\bibitem [{\citenamefont {Nishiyama}\ and\ \citenamefont
  {Yasuda}(2015)}]{Nishiyama2015}%
  \BibitemOpen
  \bibfield  {author} {\bibinfo {author} {\bibfnamefont {J.}~\bibnamefont
  {Nishiyama}}\ and\ \bibinfo {author} {\bibfnamefont {R.}~\bibnamefont
  {Yasuda}},\ }\href
  {https://doi.org/https://doi.org/10.1016/j.neuron.2015.05.043} {\bibfield
  {journal} {\bibinfo  {journal} {Neuron}\ }\textbf {\bibinfo {volume} {87}},\
  \bibinfo {pages} {63} (\bibinfo {year} {2015})}\BibitemShut {NoStop}%
\bibitem [{\citenamefont {Heine}\ \emph {et~al.}(2008)\citenamefont {Heine},
  \citenamefont {Groc}, \citenamefont {Frischknecht}, \citenamefont {Béïque},
  \citenamefont {Lounis}, \citenamefont {Rumbaugh}, \citenamefont {Huganir},
  \citenamefont {Cognet},\ and\ \citenamefont {Choquet}}]{Heine2008}%
  \BibitemOpen
  \bibfield  {author} {\bibinfo {author} {\bibfnamefont {M.}~\bibnamefont
  {Heine}}, \bibinfo {author} {\bibfnamefont {L.}~\bibnamefont {Groc}},
  \bibinfo {author} {\bibfnamefont {R.}~\bibnamefont {Frischknecht}}, \bibinfo
  {author} {\bibfnamefont {J.-C.}\ \bibnamefont {Béïque}}, \bibinfo {author}
  {\bibfnamefont {B.}~\bibnamefont {Lounis}}, \bibinfo {author} {\bibfnamefont
  {G.}~\bibnamefont {Rumbaugh}}, \bibinfo {author} {\bibfnamefont {R.~L.}\
  \bibnamefont {Huganir}}, \bibinfo {author} {\bibfnamefont {L.}~\bibnamefont
  {Cognet}},\ and\ \bibinfo {author} {\bibfnamefont {D.}~\bibnamefont
  {Choquet}},\ }\href {https://doi.org/10.1126/science.1152089} {\bibfield
  {journal} {\bibinfo  {journal} {Science}\ }\textbf {\bibinfo {volume}
  {320}},\ \bibinfo {pages} {201} (\bibinfo {year} {2008})}\BibitemShut
  {NoStop}%
\bibitem [{\citenamefont {Frost}\ \emph {et~al.}(2012)\citenamefont {Frost},
  \citenamefont {Lu},\ and\ \citenamefont {Blanpied}}]{Blanpied2012}%
  \BibitemOpen
  \bibfield  {author} {\bibinfo {author} {\bibfnamefont {N.~A.}\ \bibnamefont
  {Frost}}, \bibinfo {author} {\bibfnamefont {H.~E.}\ \bibnamefont {Lu}},\ and\
  \bibinfo {author} {\bibfnamefont {T.~A.}\ \bibnamefont {Blanpied}},\ }\href
  {https://doi.org/10.1371/journal.pone.0036751} {\bibfield  {journal}
  {\bibinfo  {journal} {PLoS ONE}\ }\textbf {\bibinfo {volume} {7}},\ \bibinfo
  {pages} {e36751} (\bibinfo {year} {2012})}\BibitemShut {NoStop}%
\bibitem [{\citenamefont {Frost}\ \emph {et~al.}(2010)\citenamefont {Frost},
  \citenamefont {Shroff}, \citenamefont {Kong}, \citenamefont {Betzig},\ and\
  \citenamefont {Blanpied}}]{Blanpied2010}%
  \BibitemOpen
  \bibfield  {author} {\bibinfo {author} {\bibfnamefont {N.~A.}\ \bibnamefont
  {Frost}}, \bibinfo {author} {\bibfnamefont {H.}~\bibnamefont {Shroff}},
  \bibinfo {author} {\bibfnamefont {H.}~\bibnamefont {Kong}}, \bibinfo {author}
  {\bibfnamefont {E.}~\bibnamefont {Betzig}},\ and\ \bibinfo {author}
  {\bibfnamefont {T.~A.}\ \bibnamefont {Blanpied}},\ }\bibfield  {booktitle}
  {\emph {\bibinfo {booktitle} {Neuron}},\ }\href
  {https://doi.org/10.1016/j.neuron.2010.05.026} {\bibfield  {journal}
  {\bibinfo  {journal} {Neuron}\ }\textbf {\bibinfo {volume} {67}},\ \bibinfo
  {pages} {86} (\bibinfo {year} {2010})}\BibitemShut {NoStop}%
\bibitem [{\citenamefont {Holcman}\ and\ \citenamefont
  {Yuste}(2015)}]{Holcman2015}%
  \BibitemOpen
  \bibfield  {author} {\bibinfo {author} {\bibfnamefont {D.}~\bibnamefont
  {Holcman}}\ and\ \bibinfo {author} {\bibfnamefont {R.}~\bibnamefont
  {Yuste}},\ }\href {https://doi.org/10.1038/nrn4022} {\bibfield  {journal}
  {\bibinfo  {journal} {Nature Reviews Neuroscience}\ }\textbf {\bibinfo
  {volume} {16}},\ \bibinfo {pages} {685} (\bibinfo {year} {2015})}\BibitemShut
  {NoStop}%
\bibitem [{\citenamefont {Jaskolski}\ \emph {et~al.}(2009)\citenamefont
  {Jaskolski}, \citenamefont {Mayo-Martin}, \citenamefont {Jane},\ and\
  \citenamefont {Henley}}]{Jaskolski2009}%
  \BibitemOpen
  \bibfield  {author} {\bibinfo {author} {\bibfnamefont {F.}~\bibnamefont
  {Jaskolski}}, \bibinfo {author} {\bibfnamefont {B.}~\bibnamefont
  {Mayo-Martin}}, \bibinfo {author} {\bibfnamefont {D.}~\bibnamefont {Jane}},\
  and\ \bibinfo {author} {\bibfnamefont {J.~M.}\ \bibnamefont {Henley}},\
  }\href {https://doi.org/10.1074/jbc.M808401200} {\bibfield  {journal}
  {\bibinfo  {journal} {Journal of Biological Chemistry}\ }\textbf {\bibinfo
  {volume} {284}},\ \bibinfo {pages} {12491} (\bibinfo {year}
  {2009})}\BibitemShut {NoStop}%
\bibitem [{\citenamefont {Tønnesen}\ \emph {et~al.}(2014)\citenamefont
  {Tønnesen}, \citenamefont {Katona}, \citenamefont {Rózsa},\ and\
  \citenamefont {Nägerl}}]{Toennesen2014}%
  \BibitemOpen
  \bibfield  {author} {\bibinfo {author} {\bibfnamefont {J.}~\bibnamefont
  {Tønnesen}}, \bibinfo {author} {\bibfnamefont {G.}~\bibnamefont {Katona}},
  \bibinfo {author} {\bibfnamefont {B.}~\bibnamefont {Rózsa}},\ and\ \bibinfo
  {author} {\bibfnamefont {U.~V.}\ \bibnamefont {Nägerl}},\ }\href
  {https://doi.org/10.1038/nn.3682} {\bibfield  {journal} {\bibinfo  {journal}
  {Nature Neuroscience}\ }\textbf {\bibinfo {volume} {17}},\ \bibinfo {pages}
  {678} (\bibinfo {year} {2014})}\BibitemShut {NoStop}%
\bibitem [{\citenamefont {Groc}\ and\ \citenamefont
  {Choquet}(2020)}]{Groc2020}%
  \BibitemOpen
  \bibfield  {author} {\bibinfo {author} {\bibfnamefont {L.}~\bibnamefont
  {Groc}}\ and\ \bibinfo {author} {\bibfnamefont {D.}~\bibnamefont {Choquet}},\
  }\href {https://doi.org/10.1126/science.aay4631} {\bibfield  {journal}
  {\bibinfo  {journal} {Science}\ }\textbf {\bibinfo {volume} {368}},\ \bibinfo
  {pages} {eaay4631} (\bibinfo {year} {2020})},\ \Eprint
  {https://arxiv.org/abs/https://www.science.org/doi/pdf/10.1126/science.aay4631}
  {https://www.science.org/doi/pdf/10.1126/science.aay4631} \BibitemShut
  {NoStop}%
\bibitem [{\citenamefont {Zeng}\ \emph {et~al.}(2016)\citenamefont {Zeng},
  \citenamefont {Shang}, \citenamefont {Araki}, \citenamefont {Guo},
  \citenamefont {Huganir},\ and\ \citenamefont {Zhang}}]{Zeng2016}%
  \BibitemOpen
  \bibfield  {author} {\bibinfo {author} {\bibfnamefont {M.}~\bibnamefont
  {Zeng}}, \bibinfo {author} {\bibfnamefont {Y.}~\bibnamefont {Shang}},
  \bibinfo {author} {\bibfnamefont {Y.}~\bibnamefont {Araki}}, \bibinfo
  {author} {\bibfnamefont {T.}~\bibnamefont {Guo}}, \bibinfo {author}
  {\bibfnamefont {R.~L.}\ \bibnamefont {Huganir}},\ and\ \bibinfo {author}
  {\bibfnamefont {M.}~\bibnamefont {Zhang}},\ }\href
  {https://doi.org/10.1016/j.cell.2016.07.008} {\bibfield  {journal} {\bibinfo
  {journal} {Cell}\ }\textbf {\bibinfo {volume} {166}},\ \bibinfo {pages}
  {1163} (\bibinfo {year} {2016})}\BibitemShut {NoStop}%
\bibitem [{\citenamefont {Zeng}\ \emph {et~al.}(2019)\citenamefont {Zeng},
  \citenamefont {Díaz-Alonso}, \citenamefont {Ye}, \citenamefont {Chen},
  \citenamefont {Xu}, \citenamefont {Ji}, \citenamefont {Nicoll},\ and\
  \citenamefont {Zhang}}]{Zeng2019}%
  \BibitemOpen
  \bibfield  {author} {\bibinfo {author} {\bibfnamefont {M.}~\bibnamefont
  {Zeng}}, \bibinfo {author} {\bibfnamefont {J.}~\bibnamefont {Díaz-Alonso}},
  \bibinfo {author} {\bibfnamefont {F.}~\bibnamefont {Ye}}, \bibinfo {author}
  {\bibfnamefont {X.}~\bibnamefont {Chen}}, \bibinfo {author} {\bibfnamefont
  {J.}~\bibnamefont {Xu}}, \bibinfo {author} {\bibfnamefont {Z.}~\bibnamefont
  {Ji}}, \bibinfo {author} {\bibfnamefont {R.~A.}\ \bibnamefont {Nicoll}},\
  and\ \bibinfo {author} {\bibfnamefont {M.}~\bibnamefont {Zhang}},\ }\href
  {https://doi.org/https://doi.org/10.1016/j.neuron.2019.08.001} {\bibfield
  {journal} {\bibinfo  {journal} {Neuron}\ }\textbf {\bibinfo {volume} {104}},\
  \bibinfo {pages} {529} (\bibinfo {year} {2019})}\BibitemShut {NoStop}%
\bibitem [{\citenamefont {Ross}(1996)}]{Ross1996}%
  \BibitemOpen
  \bibfield  {author} {\bibinfo {author} {\bibfnamefont {S.~M.}\ \bibnamefont
  {Ross}},\ }\href@noop {} {\emph {\bibinfo {title} {Stochastic Processes}}}\
  (\bibinfo  {publisher} {Wiley},\ \bibinfo {year} {1996})\BibitemShut
  {NoStop}%
\bibitem [{\citenamefont {Gardiner}(1985)}]{Gardiner1985}%
  \BibitemOpen
  \bibfield  {author} {\bibinfo {author} {\bibfnamefont {C.~W.}\ \bibnamefont
  {Gardiner}},\ }\href@noop {} {\emph {\bibinfo {title} {Handbook of stochastic
  methods}}},\ Series in Synergetics\ (\bibinfo  {publisher} {Springer},\
  \bibinfo {year} {1985})\BibitemShut {NoStop}%
\bibitem [{\citenamefont {Reichl}(1997)}]{Reichl1997}%
  \BibitemOpen
  \bibfield  {author} {\bibinfo {author} {\bibfnamefont {L.~E.}\ \bibnamefont
  {Reichl}},\ }\href@noop {} {\emph {\bibinfo {title} {A Modern Course in
  Statistical Physics}}}\ (\bibinfo  {publisher} {Jon Wiley and Sons Inc.},\
  \bibinfo {year} {1997})\BibitemShut {NoStop}%
\bibitem [{\citenamefont {Oksendal}(2000)}]{Oksendal2000}%
  \BibitemOpen
  \bibfield  {author} {\bibinfo {author} {\bibfnamefont {B.}~\bibnamefont
  {Oksendal}},\ }\href@noop {} {\emph {\bibinfo {title} {Stochastic
  Differential Equations: An Introduction}}}\ (\bibinfo  {publisher}
  {Springer},\ \bibinfo {year} {2000})\BibitemShut {NoStop}%
\bibitem [{\citenamefont {Höfling}\ and\ \citenamefont
  {Franosch}(2013)}]{Hoefling2013}%
  \BibitemOpen
  \bibfield  {author} {\bibinfo {author} {\bibfnamefont {F.}~\bibnamefont
  {Höfling}}\ and\ \bibinfo {author} {\bibfnamefont {T.}~\bibnamefont
  {Franosch}},\ }\href@noop {} {\bibfield  {journal} {\bibinfo  {journal}
  {Reports on progress in physics. Physical Society (Great Britain)}\ }\textbf
  {\bibinfo {volume} {76}},\ \bibinfo {pages} {046602} (\bibinfo {year}
  {2013})}\BibitemShut {NoStop}%
\bibitem [{\citenamefont {Fanelli}\ and\ \citenamefont
  {McKane}(2010)}]{Fanelli2010}%
  \BibitemOpen
  \bibfield  {author} {\bibinfo {author} {\bibfnamefont {D.}~\bibnamefont
  {Fanelli}}\ and\ \bibinfo {author} {\bibfnamefont {A.~J.}\ \bibnamefont
  {McKane}},\ }\href {https://doi.org/10.1103/PhysRevE.82.021113} {\bibfield
  {journal} {\bibinfo  {journal} {Phys. Rev. E}\ }\textbf {\bibinfo {volume}
  {82}},\ \bibinfo {pages} {021113} (\bibinfo {year} {2010})}\BibitemShut
  {NoStop}%
\bibitem [{\citenamefont {Takenawa}\ and\ \citenamefont
  {Suetsugu}(2007)}]{Takenawa2007}%
  \BibitemOpen
  \bibfield  {author} {\bibinfo {author} {\bibfnamefont {T.}~\bibnamefont
  {Takenawa}}\ and\ \bibinfo {author} {\bibfnamefont {S.}~\bibnamefont
  {Suetsugu}},\ }\href {https://doi.org/10.1038/nrm2069} {\bibfield  {journal}
  {\bibinfo  {journal} {Nature Reviews Molecular Cell Biology}\ }\textbf
  {\bibinfo {volume} {8}},\ \bibinfo {pages} {37} (\bibinfo {year}
  {2007})}\BibitemShut {NoStop}%
\bibitem [{\citenamefont {Pressley}(2001)}]{Pressley2001}%
  \BibitemOpen
  \bibfield  {author} {\bibinfo {author} {\bibfnamefont {A.}~\bibnamefont
  {Pressley}},\ }\href {https://books.google.com/books?id=UXPyquQaO6EC} {\emph
  {\bibinfo {title} {Elementary Differential Geometry}}}\ (\bibinfo
  {publisher} {Springer},\ \bibinfo {year} {2001})\BibitemShut {NoStop}%
\bibitem [{\citenamefont {Abraham}\ \emph {et~al.}(1988)\citenamefont
  {Abraham}, \citenamefont {Marsden},\ and\ \citenamefont
  {Rațiu}}]{Abraham1988}%
  \BibitemOpen
  \bibfield  {author} {\bibinfo {author} {\bibfnamefont {R.}~\bibnamefont
  {Abraham}}, \bibinfo {author} {\bibfnamefont {J.}~\bibnamefont {Marsden}},\
  and\ \bibinfo {author} {\bibfnamefont {T.}~\bibnamefont {Rațiu}},\ }\href
  {https://books.google.com/books?id=dWHet_zgyCAC} {\emph {\bibinfo {title}
  {Manifolds, Tensor Analysis, and Applications}}},\ \bibinfo {number} {v. 75}\
  (\bibinfo  {publisher} {Springer New York},\ \bibinfo {year}
  {1988})\BibitemShut {NoStop}%
\bibitem [{\citenamefont {Singer}(2006)}]{Singer2006}%
  \BibitemOpen
  \bibfield  {author} {\bibinfo {author} {\bibfnamefont {A.}~\bibnamefont
  {Singer}},\ }\href
  {https://doi.org/https://doi.org/10.1016/j.acha.2006.03.004} {\bibfield
  {journal} {\bibinfo  {journal} {Applied and Computational Harmonic Analysis}\
  }\textbf {\bibinfo {volume} {21}},\ \bibinfo {pages} {128} (\bibinfo {year}
  {2006})},\ \bibinfo {note} {special Issue: Diffusion Maps and
  Wavelets}\BibitemShut {NoStop}%
\bibitem [{\citenamefont {Belkin}\ and\ \citenamefont
  {Niyogi}(2008)}]{Belkin2008}%
  \BibitemOpen
  \bibfield  {author} {\bibinfo {author} {\bibfnamefont {M.}~\bibnamefont
  {Belkin}}\ and\ \bibinfo {author} {\bibfnamefont {P.}~\bibnamefont
  {Niyogi}},\ }\href
  {https://doi.org/https://doi.org/10.1016/j.jcss.2007.08.006} {\bibfield
  {journal} {\bibinfo  {journal} {Journal of Computer and System Sciences}\
  }\textbf {\bibinfo {volume} {74}},\ \bibinfo {pages} {1289} (\bibinfo {year}
  {2008})},\ \bibinfo {note} {learning Theory 2005}\BibitemShut {NoStop}%
\bibitem [{\citenamefont {Latorre}\ \emph {et~al.}(2011)\citenamefont
  {Latorre}, \citenamefont {Metzner}, \citenamefont {Hartmann},\ and\
  \citenamefont {Sch{\"u}tte}}]{Schutte2011}%
  \BibitemOpen
  \bibfield  {author} {\bibinfo {author} {\bibfnamefont {J.~C.}\ \bibnamefont
  {Latorre}}, \bibinfo {author} {\bibfnamefont {P.}~\bibnamefont {Metzner}},
  \bibinfo {author} {\bibfnamefont {C.}~\bibnamefont {Hartmann}},\ and\
  \bibinfo {author} {\bibfnamefont {C.}~\bibnamefont {Sch{\"u}tte}},\
  }\href@noop {} {\bibfield  {journal} {\bibinfo  {journal} {Commun. Math.
  Sci.}\ }\textbf {\bibinfo {volume} {9}},\ \bibinfo {pages} {1051} (\bibinfo
  {year} {2011})}\BibitemShut {NoStop}%
\bibitem [{\citenamefont {WANG}\ \emph {et~al.}(2003)\citenamefont {WANG},
  \citenamefont {PESKIN},\ and\ \citenamefont {ELSTON}}]{Elston2003}%
  \BibitemOpen
  \bibfield  {author} {\bibinfo {author} {\bibfnamefont {H.}~\bibnamefont
  {WANG}}, \bibinfo {author} {\bibfnamefont {C.~S.}\ \bibnamefont {PESKIN}},\
  and\ \bibinfo {author} {\bibfnamefont {T.~C.}\ \bibnamefont {ELSTON}},\
  }\href {https://doi.org/https://doi.org/10.1006/jtbi.2003.3200} {\bibfield
  {journal} {\bibinfo  {journal} {Journal of Theoretical Biology}\ }\textbf
  {\bibinfo {volume} {221}},\ \bibinfo {pages} {491} (\bibinfo {year}
  {2003})}\BibitemShut {NoStop}%
\bibitem [{\citenamefont {Müller}(1966)}]{Mueller1966}%
  \BibitemOpen
  \bibfield  {author} {\bibinfo {author} {\bibfnamefont {C.}~\bibnamefont
  {Müller}},\ }\href {https://www.springer.com/gp/book/9783540036005} {\emph
  {\bibinfo {title} {Spherical Harmonics}}}\ (\bibinfo  {publisher}
  {Springer},\ \bibinfo {year} {1966})\BibitemShut {NoStop}%
\bibitem [{\citenamefont {Atkinson}\ and\ \citenamefont
  {Han}(2010)}]{HandBookSphericalHarmonics2010}%
  \BibitemOpen
  \bibfield  {author} {\bibinfo {author} {\bibfnamefont {K.}~\bibnamefont
  {Atkinson}}\ and\ \bibinfo {author} {\bibfnamefont {W.}~\bibnamefont {Han}},\
  }\href {https://www.springer.com/gp/book/9783642259821} {\emph {\bibinfo
  {title} {Spherical Harmonics and Approximations on the Unit Sphere: An
  Introduction}}}\ (\bibinfo  {publisher} {Springer},\ \bibinfo {year}
  {2010})\BibitemShut {NoStop}%
\bibitem [{\citenamefont {Hugel}\ \emph {et~al.}(2009)\citenamefont {Hugel},
  \citenamefont {Abegg}, \citenamefont {de~Paola}, \citenamefont {Caroni},
  \citenamefont {Gähwiler},\ and\ \citenamefont {McKinney}}]{Hugel2009}%
  \BibitemOpen
  \bibfield  {author} {\bibinfo {author} {\bibfnamefont {S.}~\bibnamefont
  {Hugel}}, \bibinfo {author} {\bibfnamefont {M.}~\bibnamefont {Abegg}},
  \bibinfo {author} {\bibfnamefont {V.}~\bibnamefont {de~Paola}}, \bibinfo
  {author} {\bibfnamefont {P.}~\bibnamefont {Caroni}}, \bibinfo {author}
  {\bibfnamefont {B.~H.}\ \bibnamefont {Gähwiler}},\ and\ \bibinfo {author}
  {\bibfnamefont {R.~A.}\ \bibnamefont {McKinney}},\ }\href
  {https://doi.org/10.1093/cercor/bhn118} {\bibfield  {journal} {\bibinfo
  {journal} {Cerebral Cortex}\ }\textbf {\bibinfo {volume} {19}},\ \bibinfo
  {pages} {697} (\bibinfo {year} {2009})}\BibitemShut {NoStop}%
\bibitem [{\citenamefont {Blinn}(1982)}]{Blinn1982}%
  \BibitemOpen
  \bibfield  {author} {\bibinfo {author} {\bibfnamefont {J.~F.}\ \bibnamefont
  {Blinn}},\ }\href {https://doi.org/10.1145/357306.357310} {\bibfield
  {journal} {\bibinfo  {journal} {ACM Trans. Graph.}\ }\textbf {\bibinfo
  {volume} {1}},\ \bibinfo {pages} {235–256} (\bibinfo {year}
  {1982})}\BibitemShut {NoStop}%
\bibitem [{\citenamefont {Turing}(1952)}]{Turing1952}%
  \BibitemOpen
  \bibfield  {author} {\bibinfo {author} {\bibfnamefont {A.~M.}\ \bibnamefont
  {Turing}},\ }\href {http://www.jstor.org/stable/92463} {\bibfield  {journal}
  {\bibinfo  {journal} {Philosophical Transactions of the Royal Society of
  London. Series B, Biological Sciences}\ }\textbf {\bibinfo {volume} {237}},\
  \bibinfo {pages} {37} (\bibinfo {year} {1952})}\BibitemShut {NoStop}%
\bibitem [{\citenamefont {Britton}(2005)}]{Britton2005}%
  \BibitemOpen
  \bibfield  {author} {\bibinfo {author} {\bibfnamefont {N.}~\bibnamefont
  {Britton}},\ }\href {https://books.google.com/books?id=9jP8CyfC4dsC} {\emph
  {\bibinfo {title} {Essential Mathematical Biology}}},\ Springer Undergraduate
  Mathematics Series\ (\bibinfo  {publisher} {Springer London},\ \bibinfo
  {year} {2005})\BibitemShut {NoStop}%
\bibitem [{\citenamefont {Murray}(2013)}]{Murray2013}%
  \BibitemOpen
  \bibfield  {author} {\bibinfo {author} {\bibfnamefont {J.}~\bibnamefont
  {Murray}},\ }\href {https://books.google.com/books?id=JUrFoQEACAAJ} {\emph
  {\bibinfo {title} {Mathematical Biology II: Spatial Models and Biomedical
  Applications}}},\ Interdisciplinary Applied Mathematics\ (\bibinfo
  {publisher} {Springer New York},\ \bibinfo {year} {2013})\BibitemShut
  {NoStop}%
\bibitem [{\citenamefont {Gray}\ and\ \citenamefont {Scott}(1984)}]{Gray1984}%
  \BibitemOpen
  \bibfield  {author} {\bibinfo {author} {\bibfnamefont {P.}~\bibnamefont
  {Gray}}\ and\ \bibinfo {author} {\bibfnamefont {S.}~\bibnamefont {Scott}},\
  }\href@noop {} {\bibfield  {journal} {\bibinfo  {journal} {Chemical
  Engineering Science}\ }\textbf {\bibinfo {volume} {39}},\ \bibinfo {pages}
  {1087} (\bibinfo {year} {1984})}\BibitemShut {NoStop}%
\bibitem [{\citenamefont {Atzberger}(2010)}]{Atzberger_RD_2010}%
  \BibitemOpen
  \bibfield  {author} {\bibinfo {author} {\bibfnamefont {P.~J.}\ \bibnamefont
  {Atzberger}},\ }\href
  {https://doi.org/https://doi.org/10.1016/j.jcp.2010.01.012} {\bibfield
  {journal} {\bibinfo  {journal} {Journal of Computational Physics}\ }\textbf
  {\bibinfo {volume} {229}},\ \bibinfo {pages} {3474} (\bibinfo {year}
  {2010})}\BibitemShut {NoStop}%
\bibitem [{\citenamefont {Pearson}(1993)}]{Pearson1993}%
  \BibitemOpen
  \bibfield  {author} {\bibinfo {author} {\bibfnamefont {J.~E.}\ \bibnamefont
  {Pearson}},\ }\href {https://doi.org/10.1126/science.261.5118.189} {\bibfield
   {journal} {\bibinfo  {journal} {Science}\ }\textbf {\bibinfo {volume}
  {261}},\ \bibinfo {pages} {189} (\bibinfo {year} {1993})}\BibitemShut
  {NoStop}%
\bibitem [{\citenamefont {Burden}\ and\ \citenamefont
  {Faires}(2010)}]{Burden2010}%
  \BibitemOpen
  \bibfield  {author} {\bibinfo {author} {\bibfnamefont {R.~L.}\ \bibnamefont
  {Burden}}\ and\ \bibinfo {author} {\bibfnamefont {D.}~\bibnamefont
  {Faires}},\ }\href@noop {} {\emph {\bibinfo {title} {Numerical Analysis}}}\
  (\bibinfo  {publisher} {Brooks/Cole Cengage Learning},\ \bibinfo {year}
  {2010})\BibitemShut {NoStop}%
\bibitem [{\citenamefont {Iserles}(2008)}]{Iserles2008}%
  \BibitemOpen
  \bibfield  {author} {\bibinfo {author} {\bibfnamefont {A.}~\bibnamefont
  {Iserles}},\ }\href {https://doi.org/10.1017/CBO9780511995569} {\emph
  {\bibinfo {title} {A First Course in the Numerical Analysis of Differential
  Equations}}},\ \bibinfo {edition} {2nd}\ ed.,\ Cambridge Texts in Applied
  Mathematics\ (\bibinfo  {publisher} {Cambridge University Press},\ \bibinfo
  {year} {2008})\BibitemShut {NoStop}%
\bibitem [{\citenamefont {Varea}\ \emph {et~al.}(1999)\citenamefont {Varea},
  \citenamefont {Arag\'on},\ and\ \citenamefont {Barrio}}]{Varea1999}%
  \BibitemOpen
  \bibfield  {author} {\bibinfo {author} {\bibfnamefont {C.}~\bibnamefont
  {Varea}}, \bibinfo {author} {\bibfnamefont {J.~L.}\ \bibnamefont
  {Arag\'on}},\ and\ \bibinfo {author} {\bibfnamefont {R.~A.}\ \bibnamefont
  {Barrio}},\ }\href {https://doi.org/10.1103/PhysRevE.60.4588} {\bibfield
  {journal} {\bibinfo  {journal} {Phys. Rev. E}\ }\textbf {\bibinfo {volume}
  {60}},\ \bibinfo {pages} {4588} (\bibinfo {year} {1999})}\BibitemShut
  {NoStop}%
\bibitem [{\citenamefont {Jamieson-Lane}\ \emph {et~al.}(2016)\citenamefont
  {Jamieson-Lane}, \citenamefont {Trinh},\ and\ \citenamefont
  {Ward}}]{JamiesonLane2016}%
  \BibitemOpen
  \bibfield  {author} {\bibinfo {author} {\bibfnamefont {A.}~\bibnamefont
  {Jamieson-Lane}}, \bibinfo {author} {\bibfnamefont {P.~H.}\ \bibnamefont
  {Trinh}},\ and\ \bibinfo {author} {\bibfnamefont {M.~J.}\ \bibnamefont
  {Ward}},\ }in\ \href@noop {} {\emph {\bibinfo {booktitle} {Mathematical and
  computational approaches in advancing modern science and engineering}}}\
  (\bibinfo  {publisher} {Springer},\ \bibinfo {year} {2016})\ pp.\ \bibinfo
  {pages} {641--651}\BibitemShut {NoStop}%
\bibitem [{\citenamefont {Hinz}\ \emph {et~al.}(2019)\citenamefont {Hinz},
  \citenamefont {van Zwieten}, \citenamefont {Möller},\ and\ \citenamefont
  {Vermolen}}]{Hinz2019}%
  \BibitemOpen
  \bibfield  {author} {\bibinfo {author} {\bibfnamefont {J.}~\bibnamefont
  {Hinz}}, \bibinfo {author} {\bibfnamefont {J.}~\bibnamefont {van Zwieten}},
  \bibinfo {author} {\bibfnamefont {M.}~\bibnamefont {Möller}},\ and\ \bibinfo
  {author} {\bibfnamefont {F.}~\bibnamefont {Vermolen}},\ }\bibfield  {journal}
  {\bibinfo  {journal} {arXiv}\ }\href
  {https://doi.org/10.48550/ARXIV.1910.12588} {10.48550/ARXIV.1910.12588}
  (\bibinfo {year} {2019})\BibitemShut {NoStop}%
\bibitem [{\citenamefont {Simons}\ and\ \citenamefont
  {Vaz}(2004)}]{Simons2004}%
  \BibitemOpen
  \bibfield  {author} {\bibinfo {author} {\bibfnamefont {K.}~\bibnamefont
  {Simons}}\ and\ \bibinfo {author} {\bibfnamefont {W.~L.~C.}\ \bibnamefont
  {Vaz}},\ }\href@noop {} {\bibfield  {journal} {\bibinfo  {journal} {Ann. Rev.
  Biophys. Biomol. Struct.}\ }\textbf {\bibinfo {volume} {33}},\ \bibinfo
  {pages} {269} (\bibinfo {year} {2004})}\BibitemShut {NoStop}%
\bibitem [{\citenamefont {Lingwood}\ and\ \citenamefont
  {Simons}(2010)}]{Lingwood2010}%
  \BibitemOpen
  \bibfield  {author} {\bibinfo {author} {\bibfnamefont {D.}~\bibnamefont
  {Lingwood}}\ and\ \bibinfo {author} {\bibfnamefont {K.}~\bibnamefont
  {Simons}},\ }\href@noop {} {\bibfield  {journal} {\bibinfo  {journal}
  {Science (New York, N.Y.)}\ }\textbf {\bibinfo {volume} {327}},\ \bibinfo
  {pages} {46} (\bibinfo {year} {2010})}\BibitemShut {NoStop}%
\bibitem [{\citenamefont {v.~Smoluchowski}(1918)}]{Smoluchowski1918}%
  \BibitemOpen
  \bibfield  {author} {\bibinfo {author} {\bibfnamefont {M.}~\bibnamefont
  {v.~Smoluchowski}},\ }\href {https://doi.org/doi:10.1515/zpch-1918-9209}
  {\bibfield  {journal} {\bibinfo  {journal} {Zeitschrift für Physikalische
  Chemie}\ }\textbf {\bibinfo {volume} {92U}},\ \bibinfo {pages} {129}
  (\bibinfo {year} {1918})}\BibitemShut {NoStop}%
\bibitem [{\citenamefont {Gillespie}\ \emph {et~al.}(2013)\citenamefont
  {Gillespie}, \citenamefont {Hellander},\ and\ \citenamefont
  {Petzold}}]{Gillespie2013}%
  \BibitemOpen
  \bibfield  {author} {\bibinfo {author} {\bibfnamefont {D.~T.}\ \bibnamefont
  {Gillespie}}, \bibinfo {author} {\bibfnamefont {A.}~\bibnamefont
  {Hellander}},\ and\ \bibinfo {author} {\bibfnamefont {L.~R.}\ \bibnamefont
  {Petzold}},\ }\href@noop {} {\bibfield  {journal} {\bibinfo  {journal} {The
  Journal of chemical physics}\ }\textbf {\bibinfo {volume} {138}},\ \bibinfo
  {pages} {05B201\_1} (\bibinfo {year} {2013})}\BibitemShut {NoStop}%
\bibitem [{\citenamefont {Isaacson}(2013)}]{Isaacson2013}%
  \BibitemOpen
  \bibfield  {author} {\bibinfo {author} {\bibfnamefont {S.~A.}\ \bibnamefont
  {Isaacson}},\ }\href@noop {} {\bibfield  {journal} {\bibinfo  {journal} {The
  Journal of chemical physics}\ }\textbf {\bibinfo {volume} {139}},\ \bibinfo
  {pages} {054101} (\bibinfo {year} {2013})}\BibitemShut {NoStop}%
\bibitem [{\citenamefont {Isaacson}\ \emph {et~al.}(2011)\citenamefont
  {Isaacson}, \citenamefont {McQueen},\ and\ \citenamefont
  {Peskin}}]{IsaacsonPeskin2011}%
  \BibitemOpen
  \bibfield  {author} {\bibinfo {author} {\bibfnamefont {S.}~\bibnamefont
  {Isaacson}}, \bibinfo {author} {\bibfnamefont {D.}~\bibnamefont {McQueen}},\
  and\ \bibinfo {author} {\bibfnamefont {C.~S.}\ \bibnamefont {Peskin}},\
  }\href@noop {} {\bibfield  {journal} {\bibinfo  {journal} {Proceedings of the
  National Academy of Sciences}\ }\textbf {\bibinfo {volume} {108}},\ \bibinfo
  {pages} {3815} (\bibinfo {year} {2011})}\BibitemShut {NoStop}%
\bibitem [{\citenamefont {Hohenberg}\ and\ \citenamefont
  {Krekhov}(2015)}]{Hohenberg2015}%
  \BibitemOpen
  \bibfield  {author} {\bibinfo {author} {\bibfnamefont {P.}~\bibnamefont
  {Hohenberg}}\ and\ \bibinfo {author} {\bibfnamefont {A.}~\bibnamefont
  {Krekhov}},\ }\href
  {https://doi.org/https://doi.org/10.1016/j.physrep.2015.01.001} {\bibfield
  {journal} {\bibinfo  {journal} {Physics Reports}\ }\textbf {\bibinfo {volume}
  {572}},\ \bibinfo {pages} {1} (\bibinfo {year} {2015})},\ \bibinfo {note} {an
  introduction to the Ginzburg–Landau theory of phase transitions and
  nonequilibrium patterns}\BibitemShut {NoStop}%
\bibitem [{\citenamefont {Ginzburg}\ and\ \citenamefont
  {Landau}(1950)}]{Ginzburg1950}%
  \BibitemOpen
  \bibfield  {author} {\bibinfo {author} {\bibfnamefont {V.~L.}\ \bibnamefont
  {Ginzburg}}\ and\ \bibinfo {author} {\bibfnamefont {L.~D.}\ \bibnamefont
  {Landau}},\ }\href@noop {} {\bibfield  {journal} {\bibinfo  {journal} {Zh.
  Eksp. Teor. Fiz.}\ }\textbf {\bibinfo {volume} {20}},\ \bibinfo {pages}
  {1064} (\bibinfo {year} {1950})}\BibitemShut {NoStop}%
\bibitem [{\citenamefont {Allen}\ and\ \citenamefont
  {Cahn}(1979)}]{AllenCahn1979}%
  \BibitemOpen
  \bibfield  {author} {\bibinfo {author} {\bibfnamefont {S.~M.}\ \bibnamefont
  {Allen}}\ and\ \bibinfo {author} {\bibfnamefont {J.~W.}\ \bibnamefont
  {Cahn}},\ }\href
  {https://doi.org/https://doi.org/10.1016/0001-6160(79)90196-2} {\bibfield
  {journal} {\bibinfo  {journal} {Acta Metallurgica}\ }\textbf {\bibinfo
  {volume} {27}},\ \bibinfo {pages} {1085} (\bibinfo {year}
  {1979})}\BibitemShut {NoStop}%
\bibitem [{\citenamefont {Cahn}\ and\ \citenamefont
  {Hilliard}(1958)}]{Cahn1958}%
  \BibitemOpen
  \bibfield  {author} {\bibinfo {author} {\bibfnamefont {J.~W.}\ \bibnamefont
  {Cahn}}\ and\ \bibinfo {author} {\bibfnamefont {J.~E.}\ \bibnamefont
  {Hilliard}},\ }\href {https://doi.org/10.1063/1.1744102} {\bibfield
  {journal} {\bibinfo  {journal} {The Journal of Chemical Physics}\ }\textbf
  {\bibinfo {volume} {28}},\ \bibinfo {pages} {258} (\bibinfo {year} {1958})},\
  \Eprint {https://arxiv.org/abs/https://doi.org/10.1063/1.1744102}
  {https://doi.org/10.1063/1.1744102} \BibitemShut {NoStop}%
\bibitem [{\citenamefont {Gross}\ \emph {et~al.}(2020)\citenamefont {Gross},
  \citenamefont {Trask}, \citenamefont {Kuberry},\ and\ \citenamefont
  {Atzberger}}]{Atzberger_GMLS_Hydrodynamics_2020}%
  \BibitemOpen
  \bibfield  {author} {\bibinfo {author} {\bibfnamefont {B.}~\bibnamefont
  {Gross}}, \bibinfo {author} {\bibfnamefont {N.}~\bibnamefont {Trask}},
  \bibinfo {author} {\bibfnamefont {P.}~\bibnamefont {Kuberry}},\ and\ \bibinfo
  {author} {\bibfnamefont {P.}~\bibnamefont {Atzberger}},\ }\href
  {https://doi.org/https://doi.org/10.1016/j.jcp.2020.109340} {\bibfield
  {journal} {\bibinfo  {journal} {Journal of Computational Physics}\ }\textbf
  {\bibinfo {volume} {409}},\ \bibinfo {pages} {109340} (\bibinfo {year}
  {2020})}\BibitemShut {NoStop}%
\end{thebibliography}%

\appendix

\newpage
\clearpage

\section{Results on Statistics of Markov-Chains and Backward Equations}
\label{sec_fpt}

We discretize the particle drift-diffusion dynamics on the surface using
a Markov-Chain with jump rates $M_{ij}$.  The First Passage Time (FPT) 
and other statistics can be computed efficiently 
from the Markov-Chain without the need for Monte-Carlo sampling
when the discrete state space is not too large.  We show 
how results similar to the Backward-Kolmogorov PDEs~\cite{Oksendal2000}
can be obtained for our discrete Markov-Chains.

\setcounter{thm}{0}
\begin{thm}
Let $\mathbf{u}$ be a column vector with components in $i$ associated with
the statistics
\begin{equation}
u_i^{(n)} =
\Ex\left[f(X^{(N)})+\sum_{t=n}^{N-1}g(X^{(t)},t)\ \biggr\rvert\ X^{(n)}=x_i\right].
\end{equation}
The statistics $\mb{u}^{(k)}$ satisfy  
\begin{align}
\mathbf{u}^{(n-1)}&=M\mathbf{u}^{(n)}+\mathbf{g}^{(n-1)}\\
\mathbf{u}^{(N)}&=\mathbf{f}. 
\end{align}
The $M$ is the right-stochastic matrix of the Markov-Chain. 
For a given choice of functions $f,g$, 
we collect the values 
$[\mb{f}]_i = f(\mb{x}_i)$ and $[\mb{g}^{(\ell)}]_i = g(\mb{x}_i,\ell)$
as column vectors $\mb{f},\mb{g}^{(\ell)}$.
\end{thm}
\begin{proof} 
For the initial condition $X^{(n)}=x_i$, let the matrix
$P^{(m)}$ have the components $P_{ij}^{(m)}=p_j^{(m)}$.
For $m \geq n$, 
the $p_j^{(m)}$ is the solution of $\mb{p}^{(\ell+1)} = \mb{p}^{(\ell)} M$,
starting with $p_j^{(n)} = [\mb{p}]_{j}^{(n)} 
= \delta_{ij}$.  The $\delta_{ij}$ is the Kronecker 
$\delta$-function.
For $m=n$, we have $P_{ij}^{(n)}=\delta_{ij}$.
This gives
\begin{equation}
\begin{split}
u_i^{(n)}&=\Ex\left[f(X^{(N)})+\sum_{t=n}^{N-1}g(X^{(t)},t)\ 
\biggr\rvert\ X^{(n)}=x_i\right]\\
&=\sum_jP^{(N)}_{ij}f(x_j)+\sum_{t=n}^{N-1}\sum_jP_{ij}^{(t)}g(x_j,t)\\
&=\sum_j\sum_kP^{(n)}_{ik}(M^{N-n})_{kj}f(x_j)\\
&\qquad+\sum_{t=n}^{N-1}\sum_j\sum_kP_{ik}^{(n)}(M^{t-n})_{kj}g(x_j,t)\\
&=\sum_j(M^{N-n})_{ij}f(x_j)+\sum_{t=n}^{N-1}\sum_j(M^{t-n})_{ij}g(x_j,t)\\
&=\left[M^{N-n}\mathbf{f}+\sum_{t=n}^{N-1}M^{t-n}\mathbf{g}^{(t)}\right]_i \\
  &=\left[M \mb{u}^{(n+1)} + \mb{g}^{(n)}
\right]_i.
\end{split}
\end{equation}
At time $t=N$, only the term with $f$ contributes and we obtain 
$u_i^{(N)} = [\mb{f}]_i$.

\end{proof}

\begin{thm} 
Let $\mathbf{w}$ be a column vector with components associated with
the first-passage-time statistics
\begin{equation}
w_i\equiv\Ex\left[f(X^{(\tau_\Omega)}) + 
\sum_{t=0}^{\tau_{\Omega} - 1}g(X^{(t)})\ 
\biggr\rvert\ X^{(0)}=x_i\right],
\end{equation}
where $\tau_\Omega = \inf\{k \geq 0\ |\ X^{(k)}\notin\Omega\}$.
The $\mathbf{w}$ satisfies the linear equation
\begin{align}
(\hat{M} - \hat{I})\mathbf{w} &= -\mathbf{g} \\ 
\partial \mb{w} &= \mb{f}.
\end{align}
The $\partial{\mb{w}}$ extracts entries for all indices 
with $x_i \in \partial \Omega$.  The $\hat{M}$ refers to the rows 
with indices in the interior of $\Omega$.
\end{thm}

\begin{proof}
Since the equation is linear, we will first consider 
the case with $\mb{f} =0, \mb{g} \neq 0$ and then the 
case with $\mb{f} \neq0, \mb{g} = 0$.  The general 
solution is then the sum of these two cases.  For
the case with $\mb{f} =0, \mb{g} \neq 0$ we have 
\begin{equation}
\begin{split}
  w_i&=\Ex\left[\sum_{t=0}^{\tau - 1}g(X^{(t)})\ 
\biggr\rvert\ X^{(0)}=x_i\right]\\
&=\sum_{n=0}^{\infty}\Ex\left[\sum_{t=0}^{n-1}
g(X^{(t)})\ \biggr\rvert\ X^{(0)}=x_i,\tau = n\right]\Pr\{\tau=n\}\\
&=\sum_{n=0}^{\infty}g(x_i)\Pr\{\tau=n\}\\
&\quad+\sum_{n=0}^{\infty}\sum_{x_j\in\Omega}
\Pr\{X^{(1)}=x_j\ |\ X^{(0)}=x_i\}\Pr\{\tau=n\}\\
&\qquad\times\Ex\left[\sum_{t=1}^{n-1}g(X^{(t)})\ 
\biggr\rvert\ X^{(0)}=x_i,X^{(1)}=x_j,\tau=n\right]\\
&=g(x_i)+\sum_{x_j\in\Omega}M_{ij}\sum_{n=0}^{\infty}\Pr\{\tau=n\}\\
&\quad\times\Ex\left[\sum_{t=1}^{n-1}g(X^{(t)})\ 
\biggr\rvert\ X^{(0)}=x_i,X^{(1)}=x_j,\tau=n\right]\\
&=g(x_i)+\sum_{x_j\in\Omega}M_{ij}w_j.
\end{split}
\end{equation}
In matrix form,
\begin{equation}
\hat{\mathbf{w}} = \mathbf{g} + \hat{M}\mathbf{w},
\end{equation}
where $\hat{w},\hat{M}$ refers to the entries in the rows
with indices of points in the interior
of the domain $\Omega$.  This can be expressed as
\begin{equation}
(\hat{M} - \hat{I})\mathbf{w} = -\mathbf{g}.
\end{equation}
The $\hat{I}$ is the linear map that extracts entries within the 
interior of the domain $\Omega$.

We next consider the case with $\mb{f} \neq 0$ and $\mb{g} = 0$.
Similarly, this follows from 
\begin{equation}
\begin{split}
w_i&=\Ex\left[f(X^{(\tau)})\ 
\biggr\rvert\ X^{(0)}=x_i\right]\\
&=\sum_{n=0}^{\infty}\Ex\left[
f(X^{(n)})\ \biggr\rvert\ X^{(0)}=x_i,\tau = n\right]\Pr\{\tau=n\}\\
&=\sum_{n=0}^{\infty}\sum_{x_j\in\Omega}
\Pr\{X^{(1)}=x_j\ |\ X^{(0)}=x_i\}\Pr\{\tau=n\}\\
&\qquad\times\Ex\left[f(X^{(n)})\ 
\biggr\rvert\ X^{(0)}=x_i,X^{(1)}=x_j,\tau=n\right]\\
&=\sum_{x_j\in\Omega}M_{ij}\sum_{n=0}^{\infty}\Pr\{\tau=n\}\\
&\quad\times\Ex\left[f(X^{(n)})\ 
\biggr\rvert\ X^{(1)}=x_j,\tau=n\right]\\
&=\sum_{x_j\in\Omega}M_{ij}w_j.
\end{split}
\end{equation}
For $x_i \in \partial \Omega$ we have
$w_i = \Ex\left[f(X^{(\tau)})\ \biggr\rvert\ X^{(0)}=x_i\right] = f(x_i)$,
since $\tau = 0$ in this case.
In matrix form this gives $\partial\mb{w}  =  \mb{f}$ and 
\begin{equation}
(\hat{M} - \hat{I})\mathbf{w} = 0.
\end{equation}
The $\partial \mb{w}$ extracts the entries with indices $i$ 
corresponding to the boundary $\partial \Omega$.
Putting both cases together we have that $\mb{w}$ satisfies 
for general $\mb{f},\mb{g}$ the linear system
\begin{eqnarray}
\nonumber
(\hat{M} - \hat{I})\mathbf{w} &=& -\mb{g} \\
\partial\mb{w} & = & \mb{f}.
\end{eqnarray}
 
\end{proof}
These results provide approaches for computing efficiently the statistics $\mb{u}$ 
and $\mb{w}$ 
without the need in some cases for Monte-Carlo sampling when the state space is not 
too large.  These results provide an analogue of the Backward-Kolmogorov PDEs 
in the setting of our Markov-Chain discretizations of particle drift-diffusion
dynamics on curved surfaces.

\end{document}